\documentclass[11pt]{article}
\usepackage{amsmath,amssymb,amsfonts,latexsym,graphicx,amsthm}
\usepackage{fullpage,color}
\usepackage{enumerate}
\usepackage{url,hyperref}
\usepackage{comment}
\usepackage{bbm}
\usepackage{paralist}
\usepackage[ruled,linesnumbered]{algorithm2e}

\newtheorem{lemma}{Lemma}
\newtheorem{theorem}[lemma]{Theorem}
\newtheorem{defn}[lemma]{Definition}

\newtheorem{corollary}[lemma]{Corollary}

\newcommand{\ceil}[1]{\lceil #1 \rceil}

\newcommand{\wt}{\omega}
\newcommand{\NULL}{\text{null}}
\newcommand{\tree}{\widehat{T}}
\newcommand{\size}{n^{1/3}}
\newcommand{\scale}{\bar{\mu}}
\newcommand{\blowgraph}{\mathcal{G}}
\newcommand{\blowvertex}{\mathcal{V}}
\newcommand{\blowedge}{\mathcal{E}}
\newcommand{\elig}{\blowgraph_{\text{elig}}}
\newcommand{\celig}{\widehat{\blowgraph_{\text{elig}}}}
\newcommand{\gap}{12}

\newcommand{\hhat}[1]{\widehat{#1}}
\newcommand{\diff}{\mathrm{Diff}}
\newcommand{\dy}{\widehat{y}}

\newcommand{\zsum}{\mathrm{ZSum}}

\begin{document}

\title{A Scaling Algorithm for Weighted $f$-Factors in General Graphs}
\author{Ran Duan \thanks{Institute for Interdisciplinary Information Sciences, Tsinghua University,  \href{}{duanran@mail.tsinghua.edu.cn}} 
	\and Haoqing He \thanks{Institute for Interdisciplinary Information Sciences, Tsinghua University, \href{}{hehq13@mails.tsinghua.edu.cn}}
	\and Tianyi Zhang \thanks{Institute for Interdisciplinary Information Sciences, Tsinghua University, \href{}{tianyi-z16@mails.tsinghua.edu.cn}}
	}
\date{}

\maketitle

\begin{abstract}
	We study the maximum weight perfect $f$-factor problem on any general simple graph $G=(V,E,w)$ with positive integral edge weights $w$, and $n=|V|$, $m=|E|$. When we have a function $f:V\rightarrow \mathbb{N}_+$ on vertices, a perfect $f$-factor is a generalized matching so that every vertex $u$ is matched to $f(u)$ different edges. The previous best algorithms on this problem have running time $O(m f(V))$ [Gabow 2018] or $\tilde{O}(W(f(V))^{2.373}))$ [Gabow and Sankowski 2013], where $W$ is the maximum edge weight, and $f(V)=\sum_{u\in V}f(u)$.
	In this paper, we present a scaling algorithm for this problem with running time $\tilde{O}(mn^{2/3}\log W)$. Previously this bound is only known for bipartite graphs [Gabow and Tarjan 1989]. 
	The running time of our algorithm is independent of $f(V)$, and consequently it first breaks the $\Omega(mn)$ barrier
	for large $f(V)$ even for the unweighted $f$-factor problem in general graphs.
\end{abstract}

\thispagestyle{empty}
\clearpage
\pagestyle{plain}
\pagenumbering{arabic}

\section{Introduction}
Suppose we are given an undirected simple graph $G = (V, E)$ on $n$ vertices and $m$ edges, with positive integer edge weights $\wt: E\rightarrow \{1, 2, \cdots, W\}$. Let $f:V\rightarrow \mathbb{N}_+$ be a function that maps vertices to positive integers. An $f$-factor is a subset of edges $F\subseteq E$ such that $\deg_F(u) \leq f(u)$ for all $u\in V$, and $F$ is a \emph{perfect} $f$-factor if $\deg_F(u) = f(u),\forall u\in V$. In this paper we are concerned with computing a perfect $f$-factor with maximum edge weights. Note that the maximum weight $f$-factor problem can be easily reduced to the maximum weight perfect $f$-factor problem.

For polynomial running time algorithms, the previous best result on this problem has running time $\tilde{O}(m f(V))$ \cite{gabow2018data}, where conventionally $f(V) = \sum_{v\in V}f(v)$. When edge weights are small integers, a pseudo-polynomial running of $\tilde{O}(W\left(f(V)\right)^{2.373})$ was obtained using algebraic approaches by \cite{gabow2013algebraic}. For unweighted graphs, one can achieve $\tilde{O}(m\sqrt{f(V)})$ running time using algorithms from \cite{huang2017approximate,gabow1983efficient}. Faster algorithms with running time independent of $f(V)$ could be obtained previously but only in bipartite graphs: \cite{goldberg1987solving} gave a scaling algorithm that runs in time $\tilde{O}(m^{2/3}n^{5/3}\log W)$ that solves the more general min-cost unit-capacity max-flow problem. This time bound was later improved to $\tilde{O}(m\min\{n^{2/3},m^{1/2}\}\cdot \log W)$ in \cite{gabow1989faster}. For the min-cost flow problem, the running time was further improved to $\tilde{O}(mn^{1/2})$ and $\tilde{O}(m^{10/7}\log W)$ using algebraic approaches \cite{lee2014path}\cite{cohen2017negative}. If one is willing to settle for approximate solutions instead of the exact maximum, linear time algorithms can be found from \cite{huang2017approximate,duan2014linear}. A closely related problem is the min-cost perfect $b$-matching, in which every edge can be matched multiple times. There are several classical results for $b$-matchings.~\cite{gabow1989faster,edmonds1972theoretical,gabow1983scaling,gabow2018data}. Another closely related problem is minimum weight $f$-edge cover, where an $f$-edge cover is a subset of edges $F\subseteq E$ such that $\deg_F(u) \geq f(u)$ for all $u\in V$. Since the time complexity of our algorithm does not depend on $f$, it also works for the minimum weight $f$-edge cover problem.  

In this paper we prove the following result, which is the first one to break the $\Omega(mn)$ barrier of perfect $f$-factors in general graphs even for the unweighted setting.
\begin{theorem}\label{nmain}
	There is a deterministic algorithm that computes a maximum weight perfect $f$-factor in $\tilde{O}(mn^{2/3}\log W)$ time.
\end{theorem}


\subsection{Technical overview}
Our algorithm is based on the scaling approach for maximum weight matching in general graphs that runs in time $\tilde{O}(m\sqrt{n}\log W)$ from~\cite{duan2018scaling} and the blocking flow method in~\cite{even1975network,karzanov1973finding,goldberg1998beyond}. Here we begin with a sketch of our idea on finding a perfect $f$-factor in an unweighted graph. To generalize it to weighted graphs, we will adapt the scaling framework for maximum weight perfect matching from \cite{duan2018scaling}. 

The algorithm for the unweighted case uses a primal-dual approach for $f$-factors which was presented in \cite{gabow2018data,huang2017approximate}. It maintains a set of dual variables $y: V\rightarrow \mathbb{Z}$ and $z: 2^V\rightarrow \mathbb{N}$, as well as a laminar family of blossoms $\Omega\subseteq 2^V$ and a compatible $f$-factor $F$, which are initialized as $F = \Omega = \emptyset$. Basically, the algorithm invokes for $Cn^{2/3}$ times the Edmonds search procedure under an approximate complementary slackness constraint on $F, y, z, \Omega$, where $C$ is a sufficiently large constant. The key idea is that when $G$ is a simple graph, after that we wish to prove that the total deficiency of the current $f$-factor $F$ is bounded by $O(n^{2/3})$, namely $\sum_{v\in V}(f(v) - \deg_F(v))\leq O(n^{2/3})$. If this is true, then we only need extra $O(n^{2/3})$ rounds of Edmonds searches to reach a perfect $f$-factor.

Let $F^*$ be an arbitrary perfect $f$-factor. To upper bound the total deficiency $\sum_{v\in V}(f(v) - \deg_F(v))\leq O(n^{2/3})$, we need to bound the total number of edge-disjoint augmenting walks in $F^*\oplus F$. Consider any augmenting walk which is specified by a sequence of consecutive edges $(u_1, u_2), (u_2, u_3), \cdots, (u_{2s-1}, u_{2s})$, where $(u_{2i-1}, u_{2i})\in F^*, (u_{2i}, u_{2i+1})\in F$, and all $u_i$'s but $u_1, u_{2s}$ are saturated vertices in $F$ ($\deg_F(u_i) = f(u_i)$). If we start the search for $y$-values of all vertices equal to some positive constant, then $y$-value of unsaturated vertices remain equal. Since $u_1, u_{2s}$ are both unsaturated vertices, we have $y(u_1) = y(u_{2s}) = -Cn^{2/3}$.

\subparagraph*{No blossoms}
For bipartite graphs, we do not need to consider blossoms, so we can use the idea from~\cite{goldberg1998beyond,even1975network}. By approximate complementary slackness we know: $y(u_{2i-1}) + y(u_{2i})\geq -2, y(u_{2i}) + y(u_{2i+1})\leq 0$. Then we have $y(u_{2i+1}) - y(u_{2i-1})\leq 2$, $y(u_{2s-1})\geq Cn^{2/3}$. Consider the sequence of duals: $y(u_1), y(u_3), \cdots, y(u_{2s-1})$. This sequence starts with a small value $y(u_1) = -Cn^{2/3}$ but ends with a large value $y(u_{2s-1})\geq Cn^{2/3}$, and so intuitively many of the differences $y(u_{2i+1}) - y(u_{2i-1})$ should be positive. However, given the upper bound $y(u_{2i+1}) - y(u_{2i-1})\leq 2$, we would know many differences $y(u_{2i+1}) - y(u_{2i-1})$ can only belong to a very narrow range $\{1, 2\}$. In this case, since $y(u_{2i-1}) + y(u_{2i})\geq -2, y(u_{2i}) + y(u_{2i+1})\leq 0$, it must be $-1-y(u_{2i+1})\leq y(u_{2i})\leq -y(u_{2i+1})$. In words, this augmenting walk contains an edge in $V_q\times V_{-q}$, where $V_x = \{|y(u) - x|\leq 1\mid u \in V \}$, $q = y(u_{2i})$.

Since there are many different such pairs $y(u_{2i-1}), y(u_{2i+1})$, intuitively we can imagine this augmenting walk contains edges in $V_q\times V_{-q}$ for $\Omega(n^{2/3})$ different integer $q$'s. By the pigeon-hole principle, there exists one $q$ such that $|V_q\cup V_{-q}|= O(n^{1/3})$. As $G$ is a simple graph, the total number of edge disjoint augmenting walks that contains an edge in $V_q\times V_{-q}$ is at most $|V_q\cup V_{-q}|^2 = O(n^{2/3})$.

\subparagraph*{Handling blossoms}
The major difficulty for general graphs comes from the blossoms. We utilize the blossom dissolution technique from~\cite{duan2018scaling}, but it will become much more complicated for $f$-factors. To analyze the influence of blossoms, let us divide $\Omega$ into two categories: large and small: a blossom $B\in \Omega$ is large if $|B|\geq n^{1/3}$. For small blossoms, we know by definition, the total number of edges covered under all small blossoms is bounded by $n^{4/3}$. So if $F^*\oplus F$ contains $\ge Cn^{2/3}$ augmenting walks, then most augmenting walks contain less than $\frac{n^{4/3}}{Cn^{2/3}} = O(n^{2/3})$ many such edges. To restore the argument we discussed before, we could safely remove those vertices incident on any edges belonging to small blossoms from the sequence $u_1, u_3, u_5, \cdots, u_{2s-1}$. Since $O(n^{2/3})$ would intuitively be small compared to $s$, we could still work with a very long sequence of vertices that are not removed.

As for large blossoms, we could prove that $\sum_{\text{large }B\in \Omega}z(B)\leq O(n^{4/3})$. Basically, this is because the total number of large blossoms is always bounded by $n^{2/3}$, and so each round of Edmonds search could increase this sum by at most $n^{2/3}$, and therefore the algorithm could raise $\sum_{\text{large }B\in \Omega}z(B)$ to at most $O(n^{4/3})$ during $Cn^{2/3}$ executions of Edmonds search. Once we have a good handle of the total sum $\sum_{\text{large }B\in \Omega}z(B)= O(n^{4/3})$, we could argue that the ``average influence'' of large blossoms on each augmenting walk is bounded by $O(n^{2/3})$, if $F^*\oplus F$ has more than $Cn^{2/3}$ augmenting walks.

\subsection{Structure of our paper}
In Section~\ref{sec:prelim} we define the notations and basic concepts we will use in this paper, and in Section~\ref{sec:algo} the algorithm is given, whose running time analysis is given in Section~\ref{sec:time}.

\section{Preliminaries}\label{sec:prelim}
\paragraph*{Notations}
Our input is a weighted simple graph $G = (V, E, \wt)$ and a function $f:V\rightarrow \mathbb{N}_+$. For $S\subseteq V$, define $f(S) = \sum_{v\in S}f(v)$, and let $\delta(S)$ and $\gamma(S)$ be sets of edges with exactly one endpoint and both endpoints in $S$, respectively. For any edge subset $F\subseteq E$, define $\delta_F(S) = \delta(S)\cap F$, $\wt(F) = \sum_{e\in F}\wt(e)$, and $\deg_F(u) = |F\cap \{(u, v)\in E \} |$. $F\subseteq E$ is called an $f$-factor if $\deg_F(u)\leq f(u)$ for all $u\in V$. For an $f$-factor $F$, the \textbf{deficiency} of $u$ in $F$ is defined as $f(u) - \deg_F(u)$ and $u$ is \textbf{saturated} by $F$ if $f(u)-\deg_F(u) = 0$. When all vertices are saturated, $F$ is called a perfect $f$-factor.

\paragraph*{Blowup graphs}
Instead of running on the original graph, our algorithm will be operating on an auxiliary weighted graph $\blowgraph = (\blowvertex, \blowedge, \mu)$ which is called the \textbf{blowup graph}. The blowup graph is built on the original vertex set $V$ as following.
\begin{itemize}
	\item For each $e = (u, v)\in E$, add two vertices $e_u, e_v$ to $\blowvertex$ and three edges $(u, e_u), (e_u, e_v), (e_v, v)$ to $\blowedge$. All vertices in $V$ are called \textbf{original} vertices, and the new added vertices are called \textbf{auxiliary} vertices.
	\item The weights of added edge are assigned as: $\mu(u, e_u) = \mu(v, e_v) = \wt(u, v)$, $\mu(e_u, e_v) = 0$.
	\item For the added vertices, assign $f(e_u) = f(e_v) = 1$.
\end{itemize}

\begin{lemma}\label{equiv}
	Computing maximum weight perfect $f$-factor in $G$ and $\blowgraph$ are equivalent.
\end{lemma}
\begin{proof}
	Basically, we argue there is a one-to-one correspondence between perfect $f$-factors in $G$ and perfect $f$-factors in $\blowgraph$. For any perfect $f$-factor $F$ in $G$, construct a perfect $f$-factor $F^\prime \in \blowgraph$ in the following manner.
	\begin{itemize}
		\item For each $e=(u, v)\in F$, add $(u, e_u), (e_v, v)$ to $F^\prime$.
		\item For each $e=(u, v)\notin F$, add $(e_u, e_v)$ to $F^\prime$.
	\end{itemize}
	
	It is easy to see this is a one-to-one correspondence, and $\sum_{e\in F^\prime}\mu(e) = 2\sum_{e\in F}\omega(e)$.
\end{proof}

\paragraph*{LP formulation}
Computing maximum weight perfect $f$-factors on the blowup graph $\blowgraph = (\blowvertex, \blowedge, \mu)$ can be expressed as a linear program~\cite{gabow2018data}:
$$\begin{aligned}
\text{maximize}\quad & \sum_{e\in \blowedge}\mu(e)x(e)\\
\text{subject to}\quad& \sum_{e\in\delta(v)}x(e) = f(v), \forall v\in \blowvertex\\
& \sum_{e\in \gamma(B)\cup I}x(e)\leq \left\lfloor\frac{f(B) + |I|}{2}\right\rfloor, \forall B\subseteq \blowvertex, I\subseteq \delta(B)\\
& 0\leq x(e) \leq 1, \forall e\in \blowedge
\end{aligned}$$

Here, the blossom constraint $\sum_{e\in \gamma(B)\cup I}x(e)\leq \left\lfloor\frac{f(B)+|I|}{2}\right\rfloor$ is a generalization of blossom constraint $\sum_{e\in \gamma(B)}x(e)\leq \left\lfloor|B|/2\right\rfloor$ in ordinary matching. Its dual LP is written as the following.
$$\begin{aligned}
\text{minimize}\quad &\sum_{v\in \blowvertex}f(v)y(v) + \sum_{B\subseteq \blowvertex, I\subseteq \delta(B)} \left\lfloor\frac{f(B) + |I|}{2}\right\rfloor z(B, I) + \sum_{e\in \blowedge}u(e)\\
\text{subject to}\quad &yz(e) + u(e)\geq \mu(e), \forall e\in \blowedge\\
\quad &z(B, I)\geq 0, u(e)\geq 0
\end{aligned}$$
Here $yz(u, v)$ is defined as: $$yz(u,v) = y(u) + y(v) + \sum_{B, I:(u, v)\in \gamma(B)\cup I,I\subseteq \delta(B)}z(B, I)$$

\paragraph*{Blossoms}
We follow the definitions and the terminology of \cite{gabow2018data,huang2017approximate} for $f$-factor blossoms. A blossom is specified by a tuple $(B, \blowedge_B, \beta(B), \eta(B))$, where $B\subseteq \blowvertex$ is a subset of vertices, $\blowedge_B \subseteq \blowedge$ a subset of edges, $\beta(B)\in B$ a special vertex which is called the \textbf{base}, and $\eta(B)$ is either $\NULL$ or an edge from $\delta(\beta(B))\cap \delta(B)$. Blossoms follow an inductive definition below.

\begin{defn}[Blossom, \cite{gabow2018data,huang2017approximate}]
	A single vertex $v$ forms a trivial blossom, also called a singleton. Here $B = \{v\}$, $\blowedge_B = \emptyset$, $\beta(B) = v$, and $\eta(B)$ is $\NULL$. Inductively, let $B_0, B_1, \cdots, B_{l-1}$ be a sequence of disjoint singletons or nontrivial blossoms. Suppose there exists a closed walk $C_B = \{e_0, e_1, \cdots, e_{l-1}\}$ starting and ending with $B_0$ such that $e_i\in B_i\times B_{i+1}$, $(B_l = B_0)$. The vertex set $B=\bigcup_{i=0}^{l-1}B_i$ is identified as a blossom if the following are satisfied.
	\begin{enumerate}
		\item Base. If $B_0$ is a singleton, the two edges incident to $B_0$ on $C_B$, i.e., $e_0$ and $e_{l-1}$, must both be matched or both be unmatched.
		\item Alternation.  Fix a $B_i, i\neq 0$. If $B_i$ is a singleton, exactly one of $e_{i-1}$ and $e_i$ is matched. If $B_i$ is a nontrivial blossom, $\eta(B_i)= e_{i-1}$ or $e_i$.
	\end{enumerate}
	The edge set of the blossom $B$ is $\blowedge_B = C_B \cup (\cup_{i=0}^{l-1} \blowedge_{B_i})$ and its base is $\beta(B) = \beta(B_0)$. If $B_0$ is not a singleton, $\eta(B) = \eta(B_0)$. Otherwise, $\eta(B)$ may either be $\NULL$ or one edge in $\delta(B) \cap \delta(B_0)$ that is the opposite type of $e_0$ and $e_{l-1}$.
\end{defn}

A blossom is called \textbf{root blossom} if it is not contained in any other blossom. Blossoms have two different types: \textbf{light} and \textbf{heavy}. If $B_0$ is a singleton, $B$ is light/heavy if $e_0$ and $e_{l-1}$ are both unmatched/matched. Otherwise, $B$ is light/heavy if $B_0$ is light/heavy.

\begin{defn}
	Given an $f$-factor $F$, an alternating walk on $\blowgraph$ is a sequence of consecutive edges $(u_1, u_2), (u_2, u_3), \cdots, (u_{l-1}, u_l)$ such that:
	\begin{itemize}
		\item $(u_i, u_{i+1})\in \blowedge$ are different edges $1\leq i<l$.
		\item exactly one of $(u_{i-1}, u_i), (u_i, u_{i+1})$ belongs to $F$, $1<i<l$.
	\end{itemize}
	This walk is called an augmenting walk if both $(u_1, u_2), (u_{l-1}, u_l)\notin F$.
\end{defn}

When searching for an augmenting walk, a blossom behaves as a unit in the graph. These properties are formally stated by the following lemma.

\begin{lemma}[\cite{gabow2018data,huang2017approximate}]
	Let $v$ be an arbitrary vertex in $B$. There exists an even length alternating walk $P_0(v)$ and an odd length alternating walk $P_1(v)$ from $\beta(B)$ to $v$ using edges in $E_B$. Moreover, the terminal edge of $P_{0, 1}(v)$ incident to $\beta(B)$ must have a different type than $\eta(B)$, if $\eta(B)$ is defined.
\end{lemma}

We also introduce the notion of \textbf{maturity} of blossoms below.
\begin{defn}[Mature Blossom, \cite{gabow2018data,huang2017approximate}]
	A blossom is mature with respect to an $f$-factor $F$ if the following requirements are satisfied.
	\begin{enumerate}
		\item Every vertex $v\in B\setminus \{\beta(B) \}$ is saturated, namely $\deg_F(v) = f(v)$.
		\item The deficiency of $\beta(B)$ is at most $1$. Furthermore, if it is $1$, then $B$ must be a light blossom and $\eta(B)$ is $\NULL$; otherwise, $\eta(B)$ is defined.
	\end{enumerate}
\end{defn}

Our algorithm always keeps a set $\Omega$ of mature blossoms and maintains a non-negative value $z(B)$ for each $B \in \Omega$. For each blossom $B$, define a set $I(B) \subseteq \delta(B)$ which is defined as $I(B) = \delta_F(B)\oplus \{\eta(B)\}$.

\paragraph*{Augmenting and alternating paths}
To find augmentations, we need to work with the contraction graph $\widehat{\blowgraph}$ where every root blossom is contracted to a single node.

\begin{defn}[\cite{gabow2018data,huang2017approximate}]
	Let $F$, $\Omega$ and $\widehat{\blowgraph}$ be an $f$-factor, a set of blossoms and the graph obtained by contracting every root blossom in the $\Omega$. $\widehat{P} = \langle B_0, e_0, B_1, e_1, \cdots, B_l\rangle \in \widehat{\blowgraph}$ is called an augmenting path if the following requirements are satisfied.
	\begin{enumerate}
		\item The terminals $B_0$ and $B_l$ must be unsaturated singletons or unsaturated light blossoms. If $\widehat{P}$ is a closed walk ($B_0 = B_l$), $B_0$ must be a singleton and the deficiency of $\beta(B_0)$ is at least $2$. Otherwise $B_0$ and $B_l$ can be either singletons or blossoms and their deficiency must be positive.
		\item If the terminal vertex $B_0$ ($B_l$) is a singleton, then the incident terminal edges $e_0$ ($e_{l-1}$) must be unmatched. Otherwise, they can be either matched or unmatched.
		\item Let $B_i$, $0<i<l$ be an internal singleton or blossom. If $B_i$ is a singleton, then exactly one of $e_{i-1}$ and $e_i$ is matched. If $B_i$ is a nontrivial blossom, then $\eta(B_i)=e_{i-1}$ or $e_i$.
	\end{enumerate}
\end{defn}

To avoid misunderstanding, we emphasize the different between the augmenting paths and the augmenting walks. First they are defined on $\widehat{\blowgraph}$ and $\blowgraph$ respectively. Second, an augmenting walk can pass through a vertex in $\blowgraph$ several times but an augmenting path can pass through a vertex in $\widehat{\blowgraph}$ (except the endpoint) only once. In the following parts, these two concepts are used in different scenarios.

Next we define a concept of the alternating path, which is weaker than the concept of the augmenting path.

\begin{defn}
	Let $F$, $\Omega$ and $\widehat{\blowgraph}$ be an $f$-factor, a set of blossoms and the graph obtained by contracting every root blossom in the $\Omega$. A simple path $\widehat{P} = \langle B_0, e_0, B_1, e_1, \cdots, B_l\rangle$ is called an alternating path if it satisfies the following requirements.
	\begin{enumerate}
		\item The terminals $B_0$ must be unsaturated singletons or unsaturated light blossoms.
		\item If the terminal vertex $B_0$ is a singleton, then the incident terminal edges $e_0$ must be non-matching. Otherwise, they can be either matching or non-matching.
		\item For each $1\leq i<l$, if $B_i$ is a singleton, then exactly one of $e_{i-1}, e_i$ is matched. Otherwise, $\eta(B_i)=e_{i-1}$ or $e_i$.
	\end{enumerate}
\end{defn}

\paragraph*{Complementary slackness}
Throughout the algorithm, we will be maintaining an $f$-factor $F$, a set of blossoms $\Omega\subset 2^\blowvertex$, dual functions $y: \blowvertex\rightarrow \mathbb{N}$, $z:\Omega\rightarrow \mathbb{N}_{\geq 0}$ and $yz: \blowedge \rightarrow \mathbb{N}$.

For an $f$-factor $F$, we define four kinds of complementary slackness: complementary slackness, weak complementary slackness, approximate complementary slackness and weak approximate complementary slackness.

\begin{defn}\label{slack}
	In the blowup graph $\blowgraph$, an $f$-factor $F$, duals $y, z$, as well as a laminar family of blossoms $\Omega$ satisfy \textbf{complementary slackness} if the following requirements hold.
	\begin{enumerate}
		\item Dominance. For each edge $e\in \blowedge$, $yz(e)\geq \mu(e)$.
		\item Tightness. For each $e\in F$, $yz(e) = \mu(e)$.
		\item Maturity. For each blossom $B\in \Omega$, $|F\cap(\gamma(B)\cup I(B))| = \left\lfloor\frac{f(B)+|I(B)|}{2}\right\rfloor$.
	\end{enumerate}
	The \textbf{weak complementary slackness} relaxes the requirements to: 
	\begin{enumerate}
		\item Dominance. For each edge $e\in \blowedge \setminus F$, $yz(e)\geq \mu(e)$.
		\item Tightness. For each $e\in F$, $yz(e) \leq \mu(e)$.
		\item Maturity. For each blossom $B\in \Omega$, $|F\cap(\gamma(B)\cup I(B))| = \left\lfloor\frac{f(B)+|I(B)|}{2}\right\rfloor$.
	\end{enumerate}
	We emphasize the difference between these two requirements: for complementary slackness, we require dominance condition on every edge in $\blowedge$, while in the weaker version, we only need dominance on edges not in $F$.
\end{defn}

\begin{defn}\label{approx-slack}
	In the blowup graph $\blowgraph$, an $f$-factor $F$, duals $y, z$, as well as a laminar family of blossoms $\Omega$ satisfy \textbf{approximate complementary slackness} if the following requirements hold.
	\begin{enumerate}
		\item Dominance. For each edge $e\in \blowedge$, $yz(e)\geq \mu(e)-2$.
		\item Tightness. For each $e\in F$, $yz(e)\leq \mu(e)$.
		\item Maturity. For each blossom $B\in \Omega$, $|F\cap(\gamma(B)\cup I(B))| = \left\lfloor\frac{f(B)+|I(B)|}{2}\right\rfloor$.
	\end{enumerate}
	The \textbf{weak approximate complementary slackness} relaxes the requirements to: 
	\begin{enumerate}
		\item Dominance. For each edge $e\in \blowedge \setminus F$, $yz(e)\geq \mu(e)-2$.
		\item Tightness. For each $e\in F$, $yz(e)\leq \mu(e)$.
		\item Maturity. For each blossom $B\in \Omega$, $|F\cap(\gamma(B)\cup I(B))| = \left\lfloor\frac{f(B)+|I(B)|}{2}\right\rfloor$.
	\end{enumerate}
\end{defn}

\begin{lemma}[\cite{huang2017approximate}]\label{opt}
	Let $F$ be a perfect $f$-factor associated with duals $y, z$ and blossoms $\Omega$, and define perfect $F^*$ to be the maximum perfect $f$-factor. Suppose $F, \Omega, y, z$ satisfy approximate complementary slackness, then
	$$\mu(F)\geq \mu(F^*) - f(\blowvertex)$$
\end{lemma}
\begin{proof}
	We first define $u:\blowedge \to \mathbb{R}$ as
	\begin{displaymath}
	u(e) = \left\{
	\begin{array}{lr}
	\mu(e) - yz(e), &\text{if } e \in F\\
	0, &\text{otherwise} 
	\end{array}
	\right.
	\end{displaymath}
	According the approximate domination and tightness properties, we have $u(e) \ge 0$ for  all $e \in \blowedge$. Moreover, $yz(e) + u(e) \ge \mu(e) - 2$ for all $e \in \blowedge$. This gives the following:
	\begin{align*}
	\mu(F) &= \sum_{e\in F} (yz(e)+u(e)) \\
	&= \sum_{v\in V} deg_F(v)y(v) + \sum_{B\in \Omega}|F \cap (\gamma(B)\cup I(B))|z(B) + \sum_{e \in F} u(e) \\
	&= \sum_{v\in V} f(v)y(v) + \sum_{B\in \Omega}\left\lfloor\frac{f(B)+|I(B)|}{2}\right\rfloor z(B) + \sum_{e \in \blowedge} u(e) \\
	&\ge \sum_{v\in V} deg_{F^*}(v)y(v) + \sum_{B\in \Omega}|F^* \cap (\gamma(B)\cup I(B))|z(B) + \sum_{e \in F^*} u(e) \\
	&\ge \sum_{e\in F^*} (\mu(e) - 2) \\
	&\ge \mu(F^*) - f(V)\qedhere
	\end{align*}
\end{proof}

\subsection{Edmonds search}
In this subsection, we introduce two different implementations of Edmonds search. Suppose we have an $f$-factor $F$, a set of blossoms $\Omega$, and duals $y, z$ satisfying some kind of slackness condition. The purpose of Edmonds search is to reduce total deficiency of $F$ by eligible augmenting paths. We need two different notions of eligibility, namely eligibility and approximate eligibility, compatible with Definition~\ref{slack} or Definition~\ref{approx-slack}.
\begin{defn}[Eligibility, \cite{gabow2018data}]\label{elig}
	An edge $e\in \blowedge$ is eligible if $yz(e) = \mu(e)$.
\end{defn}

\begin{defn}[Approximate Eligibility, \cite{huang2017approximate}]\label{approx-elig}
	An edge $e\in E$ is approximately eligible if it satisfies one of the following.
	\begin{enumerate}
		\item $e\in \blowedge_B$ for some $B\in \Omega$.
		\item $e\notin F$ and $yz(e) = \mu(e)  -2$.
		\item $e\in F$ and $yz(e) = \mu(e)$.
	\end{enumerate}
\end{defn}
Let $\celig$ be the subgraph of $\widehat{\blowgraph}$ consisting of eligible edges. A root blossom $B^\prime\in \Omega$ is called reachable from an unsaturated root blossom $B$ via an alternating path in $\celig$, if there is an alternating path that starts at $B$ and ends at $B^\prime$. To find augmenting paths and blossoms in $\celig$, we start from any unsaturated node $u_0$ in the contraction graph $\widehat{\blowgraph}$ and grow a search tree $\tree$ rooted at $u_0$; this method was also described in \cite{gabow2018data,huang2017approximate}. All nodes in $\tree$ are classified as \textbf{outer/inner}. Initially the root is outer. Next we use a DFS-like approach to build the entire $\tree$. During the process, we keep track of a tree path $\langle u_0, e_0, u_1, \cdots, e_{l-1}, u_l\rangle$ from the root, which is guaranteed to be an alternating path. According to the type of $u$, the next edge $e_{l}$ and node $u_l$ are selected by the rules below:
\begin{enumerate}
	\item $u_l$ is outer. If $u_l$ is a singleton, then scan the next non-matching edge $e_l$ and find the other endpoint $u_{l+1}$. If $u_l$ is a nontrivial blossom, then scan the next edge $e_l$ and find the other endpoint $u_{l+1}$.
	
	\item $u_l$ is inner. If $u_l$ is a singleton, then scan the next matching edge $e_l$ and find the other endpoint $u_{l+1}$. If $u_l$ is nontrivial blossom, then assign $e_l = \eta(u_l)$ (if it was not scanned before) and find the other endpoint $u_{l+1}$.
\end{enumerate}

After finding $u_{l+1}$, we try to classify it as outer or inner: if $u_{l+1}$ is a singleton, then $u_{l+1}$ is outer if $e_l$ is matched; otherwise, $u_{l+1}$ is outer if $e_l = \eta(u_{l+1})$. Issues may arise if (1) $u_{l+1}$ was already classified by previous tree searches and there is a conflict between the new label and the old label; or (2) $u_{l+1}$ is an unsaturated then the tree search has found a new augmenting path. In either case we can construct a new blossom or reduce the total deficiency.

In the end, when all reachable singletons or root blossoms are classified as outer or inner, let $\widehat{\blowvertex}_{\text{out}}$ be the set of all outer singletons or root blossoms, and let $\widehat{\blowvertex}_{\text{in}}$ be the set of all inner singletons or root blossoms. Define $\blowvertex_{\text{out}}, \blowvertex_{\text{in}}$ to be the set of all vertices in $\blowvertex$ contained in outer and inner root blossom, respectively. Next we introduce a meta procedure that will be a basic building block, which is dual adjustment. A dual adjustment performs the following step: decrement $y(v)$ for all $v\in\blowvertex_\text{out}$, and increment $y(v)$ for all $v\in\blowvertex_\text{in}$; after that, increment by $2$ all $z(B)$ for all non-singleton root blossoms $B\in\widehat{\blowvertex}_\text{out}$, and decrement by $2$ all $z(B)$ for all non-singleton root blossoms $B\in\widehat{\blowvertex}_\text{in}$. The algorithm so far is summarized as the \textsf{AdjustDuals} algorithm~\ref{dual-adjust-no-extra}. 

\begin{algorithm}
	\caption{$\textsf{AdjustDuals}(F, \Omega, y, z)$}
	\label{dual-adjust-no-extra}
	classify every root blossom in $\Omega$ as outer or inner \;
	let $\widehat{\blowvertex_{\text{out}}}$/$\widehat{\blowvertex_{\text{in}}}$ be the set of all outer/inner root blossoms in $\celig$\;
	let $\blowvertex_{\text{out}}$/$\blowvertex_{\text{in}}$ be the set of all vertices in $\blowvertex$ contained in outer/inner root blossoms\;
	adjust the duals $y, z$ as follows:
	$$\begin{aligned}
	&y(v)\leftarrow y(v)-1, v\in \blowvertex_{\text{out}}\\
	&y(v)\leftarrow y(v)+1, v\in \blowvertex_{\text{in}}\\
	&z(B)\leftarrow z(B)+2, \text{for non-singleton root blossoms $B$ in $\widehat{\blowvertex_{\text{out}}}$}\\
	&z(B)\leftarrow z(B)-2, \text{for non-singleton root blossoms $B$ in $\widehat{\blowvertex_{\text{in}}}$}\\
	\end{aligned}$$
\end{algorithm}

\subsubsection{Bounded dual adjustments}
The first implementation of Edmonds search consists of three main steps below in the \textsf{EdmondsSearch} algorithm~\ref{bounded} \cite{gabow2018data,huang2017approximate}: (1) Augmentation and blossom formation, (2) Dual adjustment and recover, (3) Blossom dissolution. The \textsf{EdmondsSearch} algorithm requires that the y-values of all $F$ vertices have the same parity. Notice that it only performs one dual adjustment and must be used with approximate eligibility.
\begin{algorithm}
	\caption{$\textsf{EdmondsSearch}(F, \Omega, y, z)$}
	\label{bounded}
	\tcc{Precondition: unsaturated vertices must all be of the same parity}
	\tcc{Augmentation and Blossom Formation from all unsaturated vertices}
	find a maximal set $\widehat{\Psi}$ of a vertex-disjoint augmenting paths in $\celig$ and extend $\widehat{\Psi}$ to a set $\Psi$ of vertex-disjoint augmenting walks in $\elig$\;
	find a maximal set $\Omega^\prime$ of reachable mature blossoms on $\celig$\;
	update $F\leftarrow F\oplus \bigcup_{P\in \Psi}P$, $\Omega\leftarrow \Omega\cup \Omega^\prime$ and $\celig$\;
	\tcc{Dual Adjustment}
	run $\textsf{AdjustDuals}(F, \Omega, y, z)$\;
	\tcc{Recover}
	for every matching edge $(u, v)$ does not satisfy the dominance condition, choose an auxiliary node $u$, $y(u) \leftarrow \mu(u,v) - y(v) - \sum_B z(B)$\;	
	\tcc{Blossom Dissolution}
	recursively remove all root blossoms whose dual value is zero\;
\end{algorithm}

\begin{lemma}\label{boundedrunningtime}(\cite{huang2017approximate})
	The \textsf{EdmondsSearch} algorithm can be implemented in $O(m)$ time.
\end{lemma}
\begin{proof}
	First, we claim the Augmentation and blossom formation step can be implemented in $O(m)$ time.Basically we follow the paradigm of depth-first search. Iterate over all unsaturated nodes $u$ in $\celig$ and construct the search tree $\tree_u$. Add the augmenting path to $\Psi$, if any, and remove the entire tree $\tree_u$ from $\celig$ so that DFS procedure from other unsaturated nodes would avoid edges and nodes that were searched before. It is easy to see that every edge is explored at most once, so the whole DFS procedure takes $O(m)$ time.
	
	For the dual adjustment, recover and blossom dissolution steps, every edge is explored in constant time and the total running time is bounded by $O(m)$.
\end{proof}

\begin{lemma}[\cite{huang2017approximate}]
	In the \textsf{EdmondsSearch} algorithm, after augmentation and blossom formation, $\celig$ does not contain any augmenting paths.
\end{lemma}
\begin{proof}
	Suppose that, after the augmentation and blossom formation, there is an augmenting path $P$ in $\celig$. Since $\widehat{\Psi}$ is maximal, $P$ must intersect some augmenting path $P^\prime \in \widehat{\Psi}$ at a vertex $v$. However, after the augmentation and blossom formation every edge in $P^\prime$ will become ineligible, so the matching edge $(v,v^\prime) \in P$ is no longer in $\celig$, contradicting the fact that $P$ consists of eligible edges.
\end{proof}

\begin{lemma}\label{preserve1}
	The \textsf{EdmondsSearch} algorithm preserves approximate complementary slackness under the approximate eligibility definition.
\end{lemma}

\begin{proof}
	We already know from \cite{huang2017approximate} that until the recover step, weak approximate complementary slackness is preserved. So we only need to reason about the recover step. The only purpose of the recover step is to restore approximate complementary slackness from weak approximate complementary slackness. Let $y_0, z_0, yz_0$ be the duals before the recover step, which guarantees weak approximate complementary slackness. 
	
	Consider any matching edge $(u, v)$ such that $yz_0(u, v) < \mu(u, v)-2$. Assume $u$ is the auxiliary vertex that undergoes a dual change in the recover step, and let $w$ be its only neighbor such that $(u, w)$ is a non-matching edge. The dual of each vertex changes only once since every auxiliary vertex is adjacent to at most one matching edge. To argue about approximate complementary slackness, we need to verify tightness and dominance conditions on the matching edge $(u, v)$ and dominance condition on the non-matching edge $(u, w)$.
	
	For the matching edge $(u, v)$, our recover step enforces $$yz(u,v) = y(u) + y_0(v) + \sum_{B:(u,v)\in \gamma(B)\cup I(B)}z(B) = \mu(u,v)$$ So dominance and the tightness are satisfied simultaneously.
	
	For the non-matching edge $(u, w)$, by tightness condition $$ y(u) = \mu(u, v) - y_0(v) - \sum_{B:(u,v)\in \gamma(B)\cup I(B)}z(B) \ge yz_0(u,v) - y_0(v) - \sum_{B:(u,v)\in \gamma(B)\cup I(B)}z(B) = y_0(u)$$ we thus have $yz(u, w) = y(u) + y_0(w) + \sum_{B:(w,v)\in \gamma(B)\cup I(B)}z(B) \ge yz_0(u, w) \ge \mu(u, w) - 2$. Thus dominance is satisfied.
\end{proof}

\subsubsection{Unbounded dual adjustments}
The second implementation of Edmonds search is described in pseudo-code~\ref{unbounded}. It also requires that the y-values of all $F$ vertices (a subset of unsaturated vertices) have the same parity. This algorithm searches for an augmenting path only from a set $U$ of unsaturated vertices whose $y$ values share the same parity, which means the augmenting path must have at least one end in $U$ and possibly both. The search iteratively performs Blossom Formation, Dual Adjustment, and Blossom Dissolution, halting after finding an augmenting path from vertices in $U$ or making $D$ dual adjustments.

\begin{algorithm}
	\caption{$\textsf{PQ-Edmonds}(F, \Omega, y, z, U, D)$}
	\tcc{Precondition: $\{y(u)|u\in U\}$ must all be of the same parity}
	\label{unbounded}
	\While{no augmenting paths from vertices in $U$ are found, or less than $D$ dual adjustments have been made so far}{
		\tcc{Blossom Formation}
		find a maximal set $\Omega^\prime$ of reachable mature blossoms on $\celig$\;
		update $\Omega\leftarrow \Omega\cup \Omega^\prime$ and $\celig$\;
		\tcc{Dual Adjustment}
		run $\textsf{AdjustDuals}(F, \Omega, y, z)$\;
		\tcc{Blossom Dissolution}
		recursively remove all root blossoms whose dual value is zero\;
	}
	\tcc{recover complementary slackness from weak complementary slackness}
	for every matching edge $(u, v)$ does not satisfy the dominance condition, choose an auxiliary node $u$, $y(u) \leftarrow \mu(u,v) - y(v) - \sum_B z(B)$
\end{algorithm}

\begin{lemma}\label{pq-search}
	The \textsf{PQ-Edmonds} algorithm can be implemented in $O(m\log n)$ time. Moreover, the $y(u)$ for unsaturated vertex $u$ not in $U$ will not be increased during the algorithm.
\end{lemma}
\begin{proof}
	By \cite{gabow2018data}, these steps until the recover step can be implemented in $O(|\blowedge| + |\blowvertex|\log n) = O(m\log n)$ time. For the recover step, each edge can be adjusted in constant time. The total running time is bounded by $O(m\log n)$.
	
	One more remark: in the original paper \cite{gabow2018data}, their algorithm actually does not contain this parameter $U$; namely $U$ is always equal to the set of all unsaturated vertices. This slack can be remedied by the following reduction. For each unsaturated vertex  $v\notin U$, match $v$ to $f(v) - \deg_F(v)$ new temporary vertices whose duals are equal to $-y(v)$ and the matching edges have zero weight. So in the new graph $U$ becomes the set of all unsaturated vertices. During executing \cite{gabow2013algebraic}'s original algorithm on the new graph, whenever the dual of any temporary vertex is about to decrease, we can abort the algorithm and claim an augmenting path from $U$ to $\blowvertex\setminus U$, so the $y$-values for the unsaturated vertices not in $U$ are not increased.
\end{proof}

\begin{lemma}\label{preserve0}
	If $y(u), y(v)$ have the same parity and $\mu, z$ are both even, the \textsf{PQ-Edomonds} algorithm preserves the complementary slackness under the eligibility definition.
\end{lemma}
\begin{proof}
	We already know from \cite{huang2017approximate} that until the recover step, weak complementary slackness is preserved. So we only need to reason about the recover step. The only purpose of the recover step is to restore complementary slackness from weak complementary slackness. Let $y_0, z_0, yz_0$ be the duals before the recover step, which guarantees weak complementary slackness. 
	
	Consider any matching edge $(u, v)$ such that $yz_0(u, v) < \mu(u, v)$. Assume $u$ is the auxiliary vertex that undergoes a dual change in the recover step, and let $w$ be its only neighbor such that $(u, w)$ is a non-matching edge. The dual of each vertex changes only once since every auxiliary vertex is adjacent to at most one matching edge. To argue about complementary slackness, we need to verify tightness and dominance conditions on the matching edge $(u, v)$ and dominance condition on the non-matching edge $(u, w)$.
	
	For the matching edge $(u, v)$, our recover step enforces $$yz(u,v) = y(u) + y_0(v) + \sum_{B:(u,v)\in \gamma(B)\cup I(B)}z(B) = \mu(u,v)$$ So the dominance and tightness are satisfied simultaneously.
	
	For the non-matching edge $(u, w)$, by tightness condition $$ y(u) = \mu(u, v) - y_0(v) - \sum_{B:(u,v)\in \gamma(B)\cup I(B)}z(B) \ge yz_0(u,v) - y_0(v) - \sum_{B:(u,v)\in \gamma(B)\cup I(B)}z(B) = y_0(u)$$ we thus have $yz(u, w) = y(u) + y_0(w) + \sum_{B:(w,v)\in \gamma(B)\cup I(B)}z(B) \ge yz_0(u, w) \ge \mu(u, w)$. Thus dominance is satisfied.
\end{proof}

\section{The Scaling Algorithm}\label{sec:algo}
Our algorithm follows the idea of the scaling algorithm in \cite{duan2018scaling} for maximum weight perfect matching. Suppose currently we have maintained an $f$-factor $F$, along with a laminar family of blossoms $\Omega$ and duals $y, z$. Throughout the algorithm we assume $y$ always assigns integer values and $z$ always assigns even non-negative integers. For any $B \subseteq \Omega$, $B$ is called a \emph{large blossom} if $|B\cap V|\geq \size$; otherwise it is deemed a \emph{small blossom}. 

The scaling algorithm maintains an $f$-factor $F$, a family of blossoms $\Omega$, as well as duals $y, z$, and it is divided into $\lceil \log (2f(V)W)\rceil$ iterations. Let $\scale$ be the edge weight function that keeps track of the scaled edge weights in each iteration. Initially before the first iteration, assign $F, \Omega = \emptyset$, $y, z, \scale = 0$. At the beginning of each iteration, define $F_0$ to be the $f$-factor from the previous iteration. Empty the matching $F\leftarrow \emptyset$, and update weights and duals as following.
$$\begin{aligned}
\scale(e) &\leftarrow 2\left(\scale(e) + \text{the next bit of }2f(V)\mu(e)\right)\\
y(u) &\leftarrow 2y(u) + 3\\
z(B) &\leftarrow 2z(B)
\end{aligned}$$

The algorithm involves an important subroutine: the Dissolve algorithm~\ref{dissolve}.

\begin{algorithm}
	\caption{$\textsf{Dissolve}(B, y, z, \Omega)$}
	\label{dissolve}
	\For{$u\in B$ or exists $v\in B$ such that $(u, v)\in I(B)$}{
		$y(u)\leftarrow y(u) + z(B)/2$\;
	}
	$z(B)\leftarrow 0$ and remove it from $\Omega$\;
\end{algorithm}

Then we apply the Dissolve algorithm~\ref{dissolve} to dissolve every large blossom $B\in \Omega$, and repeatedly dissolve any small root blossom $B$ if $z(B)\leq \gap$. Then, reweight the graph $\scale(u, v) \leftarrow \scale(u, v) - y(u) - y(v), \forall (u, v)\in E$, and reassign $y(u)\leftarrow 0, \forall u\in V$. 

Let $B_1, B_2, \cdots, B_l$ be all the nontrivial root small blossoms in $\Omega$ that are not dissolved yet. First dissolve all blossoms. For each $1\leq i\leq l$, dissolve $B_i$ and add to $F$ the matching edge set $I_{F_0}(B_i)\setminus\{\eta(B_i) \}$. To ensure tightness on these matching edges that are newly added to $F$, for each such edge $(u, v)\in I_{F_0}(B_i)\setminus\{\eta(B_i) \}, u\in B_i, v\notin B_i$, reassign $y(v)\leftarrow \mu(u, v) - y(u)$. After that, construct a subgraph $H_i$ starting with $H_i = \blowgraph[B_i]$, and then add to $H_i$ all the endpoints of $I_{F_0}(B_i)\setminus\{\eta(B_i)\}$ along with edges in $I_{F_0}(B_i)\setminus\{\eta(B_i)\}$. Then repeat the following process until the $y$ values of all unsaturated vertices are no more than $6$. Apply the PQ-Edmonds algorithm~\ref{unbounded} under the eligibility condition in subgraph $H_i$ with edges weights $\scale$ to reduce total deficiency against function $f$ by one or perform $D$ dual adjustments, where $D$ is the gap between the largest $y$ values and the second largest $y$ values among unsaturated vertices.

After we are done with all of $B_1, B_2, \cdots, B_l$, run the EdmondsSearch algorithm~\ref{bounded} under the approximate eligibility condition on the entire graph $G$ for $\ceil{C n^{2/3}}+6$ times that would reduce the dual of each unsaturated vertex to $-\ceil{C n^{2/3}}$, where $C$ is a large constant to be determined in the end. If this is the last iteration of scaling when we have exhausted all bits of integer weights $4f(V)\mu(e)$, repeatedly apply the PQ-Edmonds algorithm~\ref{unbounded} under the approximate eligibility condition on the whole graph $G$ until the overall deficiency becomes zero. The Scaling algorithm~\ref{scaling} summarizes the algorithm so far.

\begin{algorithm}[htbp]
	\caption{$\textsf{Scaling}(V, E, \mu, f)$}
	\label{scaling}
	$y, z\leftarrow 0$, $F, \Omega\leftarrow\emptyset$\;
	\For{$\text{iter} = 1, \cdots, \lceil\log(2f(V)W)\rceil$}{
		\tcc{scaling}
		$\scale(e)\leftarrow 2\left(\scale(e) + \text{the next bit of }2f(V)\mu(e)\right)$\;
		$y(u) \leftarrow 2y(u) + 3$\;
		$z(B) \leftarrow 2z(B)$\;
		$F_0\leftarrow F$, $F\leftarrow\emptyset$\;
		\tcc{blossom dissolution(Line 7-15)}
		\While{exists a large blossom $B\in \Omega$, or a root blossom $B$ with $z(B)\leq \gap$}{
			run $\textsf{Dissolve}(B, y, z, \Omega)$\;
		}
		$\scale(u, v)\leftarrow \scale(u, v) - y(u) - y(v), \forall (u, v)\in E$\;
		$y(u)\leftarrow 0, \forall u\in V$\;
		let $B_1,B_2,\cdots, B_l\in \Omega$ be all the root small blossoms not dissolved yet\;
		\While{exists a blossom $B\in \Omega$}{
			run $\textsf{Dissolve}(B, y, z, \Omega)$\;
		}
		\tcc{augmentation within small blossoms(Line 16-27)}
		\If{$(u,v)\in I_{F_0}(B_j)\setminus\{\eta(B_j)\}$ for some previous root small blossom $B_j$}{
			$F \leftarrow F \cup \{(u,v)\}$\;
			if $u\in B_j, v\notin B_j$, $y(v) \leftarrow \scale(u,v) - y(u)$\;
		}
		\For{$i = 1, 2, \cdots, l$}{
			\While{$\max\{y(u)\mid \deg_F(u)<f(u), u\in B_i \}>6$}{
				let $Y_1, Y_2$ be the largest and second largest $y$ values of unsaturated vertices in $B_i$\;
				define $U\subseteq B_i$ to be the set of unsaturated vertices whose $y$ values equal $Y_1$\;
				define $H_i = \blowgraph[B_i] \cup$ all the endpoints of $I(B_i)\setminus\{\eta(B_i)\}$\;
				run $\textsf{PQ-Edmonds}(F, \Omega, y, z, U, Y_1-Y_2)$ in subgraph $H_i$\;
			}
		}
		\tcc{deficiency reduction}
		run $\textsf{EdmondsSearch}(F, \Omega, y, z)$ on the entire graph $G$ for $\ceil{C n^{2/3}}+6$ times\;
	}
	\tcc{weight adjustment}
	\For{an edge $(u,v)\in\blowedge$ that $yz(u,v)<\mu(u,v)$}{
		$\mu(u,v) \leftarrow yz(u,v)$\;
	}
	\tcc{PQ-deficiency reduction}
	repeat $\textsf{PQ-Edmonds}(F, \Omega, y, z, \{u\mid u\in\blowvertex\text{ unsaturated} \}, \infty)$ on the entire graph $G$ until the total deficiency becomes zero\;
\end{algorithm}

\subsection{Correctness}
We begin with some basic lemmas.

\begin{lemma}\label{blowedge}
	For any blossom $B \in \Omega$ in $\blowgraph$, the edge $e \in \delta(B)$ has the form of $(u,e_u)$ where $u \in B$ is an original vertex and $e_u$ is an auxiliary vertex.
\end{lemma}
\begin{proof}
	Let $u,v$ be the original vertices and $e_u, e_v$ be the auxiliary vertices in $\blowgraph$. By construction of the blowup graph $\blowgraph$, an auxiliary vertex $e_u$ only has degree $2$. Hence, if $e_u$ belongs to $B$, then its only two neighbors which are $u, e_v$ must also belong to $B$, and thus the edge $e \in \delta(B)$ has the form of $(u,e_u)$ where $u$ belongs to $B$.
\end{proof}

\begin{defn}
If $\eta(B)$ is not null and has the form of $(u,e_u)$ where $u \in B$ is an original vertex and $e_u$ is an auxiliary vertex, define $\zeta(B) = (e_u, e_v)$.
\end{defn}

\begin{lemma}\label{dissloveedge}
	Let $y_0,z_0$ be the original duals, then after dissloving a blossom $B$, we have:
	\begin{enumerate}
		\item For each $(u,v)\in \blowedge$, $y(u) +y(v) \ge y_0(u) + y_0(v) + \sum_{(u,v)\in \gamma(B)\cup I(B)}z_0(B)$.
		\item For each $(u,v)\in F_0$, $$y(u) +y(v) = y_0(u) + y_0(v) + \sum_{(u,v)\in \gamma(B)\cup I(B)}z_0(B) + \frac{1}{2}\sum_{(u,v)\in \{\zeta(B),\eta(B)\}}z_0(B)$$
	\end{enumerate}
\end{lemma}
\begin{proof}
	As $y(u) = y_0(u) + \frac{1}{2}\sum_{\exists w,(u,w)\in \gamma(B)\cup I(B)}z_0(B)$ and $y(v) = y_0(v) + \frac{1}{2}\sum_{\exists w,(v,w)\in \gamma(B)\cup I(B)}z_0(B)$, we have $$y(u) +y(v) \ge y_0(u) + y_0(v) + \sum_{(u,v)\in \gamma(B)\cup I(B)}z_0(B)$$

	For $(u,v) \in F_0$, if $(u,v)\in \gamma(B)\cup I(B)$, $y(u)+y(v)=y_0(u)+y_0(v)+z_0(B)$. Otherwise, assume there exists a vertex $w \neq v$ such that $(u, w) \in \gamma(B)\cup I(B)$. If $u$ belongs to $B$, as $(u,v)\in F_0$, $(u,v)$ must be $\eta(B)$ and $y(u) = y_0(u) + \frac{1}{2}z_0(B)$. If $(u,w) = \eta(B)$ but $u \notin B$, $(u,v)$ must be $\zeta(B)$ and $y(u) = y_0(u) + \frac{1}{2}z_0(B)$. Otherwise, $(u,w) \in I(B)\setminus\{ \eta(B)\}$ but $u \notin B$. In this case, $u$ is an auxiliary vertex and both $(u,v)$ and $(u,w)$ belongs to $F_0$, which is a contradiction to $f(u)=1$. So $$y(u) +y(v) = y_0(u) + y_0(v) + \sum_{(u,v)\in \gamma(B)\cup I(B)}z_0(B) + \frac{1}{2}\sum_{(u,v)\in \{\zeta(B),\eta(B)\}}z_0(B)$$.
\end{proof}

\begin{lemma}\label{scale-slack}
	If the edge weights $\scale_0$, the $f$-factor $F_0$, the duals $y_0, z_0$ and the blossoms $\Omega_0$ satisfy approximate complementary slackness at the beginning of the step of scaling, we have two properties right after the scaling step:
	\begin{enumerate}
		\item For each $e\in \blowedge$, $\scale(e)\leq yz(e)$.
		\item For each $e\in F_0$, $\scale(e)\geq yz(e) - 6$.
	\end{enumerate}
\end{lemma}
\begin{proof}
	By the approximate complementary slackness, $yz_0(e)\geq \scale_0(e) - 2$ for all $e\in \blowedge$. Since $yz(e) = 2yz_0(e) + 6$ and $2\scale_0(e) + 2 \geq \scale(e) \geq 2\scale_0(e)$ after the scaling step, $$yz(e) = 2yz_0(e) + 6\geq 2\scale_0(e) + 2 \geq \scale(e)$$.
	Similarly, for each $e\in F_0$, $yz(e) = 2yz_0(e) + 6\leq 2\scale_0(e)+6 \leq \scale(e)+6$, namely $\scale(e)\geq yz(e)-6$.
\end{proof}

\begin{lemma}\label{half}
	There are two properties right after the step of blossom dissolution:
	\begin{enumerate}
		\item For each $(u, v)\notin F_0 \cup \bigcup_{i=1}^l\gamma(B_i)\cup I(B_i)$, $\scale(u, v)\leq 0$.
		\item For each $(u, v)\in \blowedge$, $\scale(u, v)\leq 2\min\{y(u), y(v)\}$.
	\end{enumerate}
\end{lemma}
\begin{proof}
	Let $\scale_1, y_1, z_1, \Omega_1$ be the edge weights, duals and blossoms at the beginning of the step of blossom dissolution.
	Let $\Omega_1^\prime\subseteq \Omega_1$ be the set of all blossoms that are dissolved within this step of blossom dissolution (Line 7,8 in the Scaling algorithm~\ref{scaling}) before the reweighting step. By Lemma~\ref{dissloveedge} and Lemma~\ref{scale-slack}, 
	$$\begin{aligned}
	\scale(u, v) &\leq \scale_1(u, v) - y_1(u) - y_1(v) - \sum_{\substack{B\in\Omega_1^\prime \\ (u, v)\in \gamma(B)\cup I(B)}}z_1(B) \\
	&\leq yz_1(u, v) - y_1(u) - y_1(v) - \sum_{\substack{B\in\Omega_1^\prime \\ (u, v)\in \gamma(B)\cup I(B)}}z_1(B)\\
	&= \sum_{\substack{B\in\Omega_1\setminus \Omega_1^\prime \\ (u, v)\in \gamma(B)\cup I(B)}}z_1(B)
	\end{aligned}$$ The last term is zero when $(u, v)\notin \bigcup_{i=1}^l\gamma(B_i)\cup I(B_i)$.
	
	Hence, by the end of the step of blossom dissolution,
	$$\begin{aligned}
	y(u) &= \frac{1}{2}\sum_{\substack{B\in \Omega_1\setminus \Omega_1^\prime \\ \exists(u, w)\in \gamma(B)\cup I(B)}}z_1(B)\geq \frac{1}{2}\sum_{\substack{B\in \Omega_1\setminus \Omega_1^\prime \\ (u, v)\in \gamma(B)\cup I(B)}}z_1(B) \geq \frac{1}{2}\scale(u, v)
	\end{aligned}$$
	By symmetry, we can also prove $y(v)\geq \frac{1}{2}\scale(u, v)$. Then $\scale(u, v)\leq 2\min\{y(u), y(v)\}$.
\end{proof}

Next we study what happens during the step of augmentation within small blossoms. If a matching edge $(u,v)$ is newly added to $F$ in line 17, for some small blossom $B_i$, $(u,v) \in I_{F_0}(B_i)\setminus\{\eta(B_i)\}$. By Lemma~\ref{blowedge}, $(u,v)$ has the form of $(u, e_u)$ where $u \in B$.

\begin{lemma}\label{tight}
	If an matching edge $(u,e_u)$ is newly added to $F$ (in line 17),  $\scale(u,e_u) \geq y(u) + y(e_u) - 6$ right after the step of blossom dissolution.
\end{lemma}
\begin{proof}
	Let $\scale_1, y_1, z_1, \Omega_1$ be the edge weights, duals and blossoms at the beginning of the step of blossom dissolution. Let $\Omega_1^\prime\subseteq \Omega_1$ be the set of all blossoms that are dissolved within this step of blossom dissolution(Line 7,8 in the Scaling algorithm~\ref{scaling}) before the step of reweighting. As $(u,e_u) \in I_{F_0}(B)\setminus\{\eta(B)\}$, where $B \in \Omega_1\setminus \Omega_1^\prime$. Then $(u,e_u) \notin\{\zeta(B),\eta(B)\}$ for any $B \in \Omega_1$. By Lemma~\ref{dissloveedge} and Lemma~\ref{scale-slack},
	$$\begin{aligned}
	\scale(u, e_u) &= \scale_1(u,e_u) - y_1(u) - y_1(e_u) - \sum_{\substack{B\in\Omega_1^\prime \\ (u, e_u)\in \gamma(B)\cup I(B)}}z_1(B) - \frac{1}{2}\sum_{\substack{B\in\Omega_1^\prime \\ (u, e_u)\in \{\zeta(B),\eta(B)\}}}z_1(B)\\
	&\geq yz_1(u, e_u) - y_1(u) - y_1(e_u)  - \sum_{\substack{B\in\Omega_1^\prime \\ (u, e_u)\in \gamma(B)\cup I(B)}}z_1(B) - 6\\
	&= \sum_{\substack{B\in\Omega\setminus \Omega_1^\prime \\ (u, e_u)\in \gamma(B)\cup I(B)}}z_1(B) - 6
	\end{aligned}$$
	
	Also notice that
	$$\begin{aligned}
	y(u) + y(e_u) &= \sum_{\substack{B\in \Omega_1\setminus \Omega_1^\prime \\ (u, e_u)\in \gamma(B)\cup I(B)}}z_1(B) + \frac{1}{2}\sum_{\substack{B\in \Omega_1\setminus \Omega_1^\prime \\ (u, e_u)\in\{\zeta(B),\eta(B)\}}}z_1(B)\leq \scale(u, e_u) + 6
	\end{aligned}$$ which concludes the proof.
\end{proof}

Now we argue that adding $(u, e_u)$ to $F$ and reassigning $y(e_u)\leftarrow \scale(u, e_u) - y(u)$ does not harm the complementary slackness.
\begin{lemma}\label{zero}
	After we have added $(u, e_u)$ to $F$ and reassigned $y(e_u)\leftarrow \scale(u, e_u) - y(u)$, the edge weights $\scale$, the $f$-factor $F$, the duals $y, z$ and the blossoms $\Omega$ satisfy the complementary slackness. Plus, $y(e_u)\geq \frac{1}{2}\scale(u, e_u)-6$.
\end{lemma}
\begin{proof}
	Let $y_2, z_2, \Omega_2$ be the duals and blossoms right after the step of blossom dissolution. By Lemma~\ref{half}, complementary slackness is already satisfied before reassignment of $y(e_u)$. Such reassignment only affects the complementary slackness condition on edge $(e_u, e_v)$, so we only need to verify $y(e_u) + y(e_v)\geq \scale(e_u, e_v)$ after the reassignment.
	
	On the one hand, by Lemma~\ref{half}, $\scale(e_u, e_v)\leq 0$. By Lemma~\ref{tight}, $y(e_u) = \scale(u, e_u) - y_2(u)\geq y_2(e_u) - 6$, namely $y(e_u)$ decreases by at most $6$. By Lemma~\ref{half}, we conclude $y(e_u)\geq \frac{1}{2}\scale(u, e_u)-6$. Moreover, 
	
	$$y(e_u) = \scale(u, e_u) - y_2(u)\geq y_2(e_u) - 6\geq \frac{1}{2}\sum_{\substack{B\in \Omega_1\setminus \Omega_1^\prime\\ (u, e_u)\in \gamma(B)\cup I(B)}}z_1(B) - 6\geq \frac{1}{2}z_1(B_i) - 6\geq 0$$ The last inequality is because $z_1(B_i)\geq 12$ since we did not dissolve $B_i$ during the step of blossom dissolution. Similarly, we can also prove $y(e_v)\geq 0$. Hence, $y(e_u) + y(e_v)\geq 0\geq \scale(e_u, e_v)$.
\end{proof}

From Lemma~\ref{pq-search}, the $y$-values of unsaturated vertices outside current $U$ cannot increase, so we can always search from unsaturated vertices with largest $y$-values. The most important step is to analyze the behavior of Edmonds search within the subgraph $H_i$.
\begin{lemma}\label{subgraph}
	Within the while-loop, for any $(u,e_u)$ newly added to $F$ and $u\in B_i$ reachable via alternating walks from $U$, $y(u)\geq Y_1$. Plus, $y(e_u)\geq 0$ at any moment.
\end{lemma}
\begin{proof}
	We prove that for any vertex $u$, $y(u)\geq Y_1$ holds when the first time $u$ becomes reachable from $U$; later on $y(u)$ would always be at least $Y_1$ as dual adjustments always put the most decreases on duals of vertices in $U$. Suppose at some point a vertex $v$ first becomes reachable from $U$ via an eligible edge $(u, v)\in \blowedge[B_i]$. On the one hand, by Definition~\ref{approx-elig} of eligibility $\scale(u, v) \geq y(u) + y(v)\geq Y_1 + y(v)$, where the inequality holds by induction that $y(u)\geq Y_1$ as $u$ was already reachable from $U$; on the other hand, since $y(v)$ was not changed before, by Lemma~\ref{half}, $y(v)\geq \frac{1}{2}\scale(u,v)\geq \frac{1}{2}Y_1 + \frac{1}{2}y(v)$ at the time, and hence $y(v)\geq Y_1$.
	
	For the second statement, we can argue similarly with Lemma~\ref{zero} that $y(e_u)\geq Y_1 - 6$ the first time it becomes reachable from $U$. Notice that there can be at most $Y_1 - 6$ dual adjustments later on, $y(e_u)$ is at least $0$ in the end. (See line 21.)
\end{proof}

Then we can conclude the correctness of the algorithm
\begin{lemma}
	The \textsf{Scaling} algorithm~\ref{scaling} returns a maximum weight perfect $f$-factor in $G$.
\end{lemma}
\begin{proof}
	First we claim that approximate complementary slackness is maintained until the step of weight adjustment. By Lemma~\ref{half}, the tightness of complementary slackness is satisfied after the step of blossom dissolution. For each edge $e$ newly added to $F$, $yz(e) = \mu(e)$ and the dominance of complementary slackness is satisfied. For the edges in $H_i$, the \textsf{PQ-Edmonds} preserve complementary slackness by Lemma~\ref{preserve0}. For the edge not adjoining to any $H_i$, the duals does not change. For the edge $(u,v)$ which $u \in H_i$ and $v$ not belongs to any $H_j$, we have $\mu(u,v) \le 0$, $\mu(u,v) \le 2y(v)$ by Lemma~\ref{half} and $y(u) \ge 0$ by Lemma~\ref{subgraph}. The complementary slackness is maintained after the step of augmentation within small blossom. Since complementary slackness is stronger than approximate complementary slackness and the \textsf{EdmondsSearch} algorithm preserves approximate complementary slackness by Lemma~\ref{preserve1}, approximate complementary slackness is maintained at the end of each iteration.

	Let $\mu, \mu^\prime$ be the edge weight respectively before and after the step of weight adjustment on line 30-32. As $y,z,\mu,\Omega$ satisfy approximate complementary slackness,  $y,z,\mu^\prime,\Omega$ satisfy complementary slackness, we know for each edge $e$, $\mu(e) - \mu'(e) \in [0,2]$. Since \textsf{PQ-Edmonds} algorithm preserves complementary slackness by Lemma~\ref{preserve0}, complementary slackness is maintained with respect to edge weights $\mu^\prime$ after the Algorithm~\ref{scaling}. Now, again by $\mu(e) - \mu'(e) \in [0,2], \forall e\in \blowedge$, we know, $y, z, F, \Omega$ still satisfies approximate complementary slackness with respect to $\mu$ after Algorithm~\ref{scaling} is completed.
	
	After the step of PQ-deficiency reduction, the total deficiency becomes zero. Then, according to Lemma~\ref{opt}, $\mu(F^*) - \mu(F)\leq  f(\blowvertex)$. Since for every edge $e\in \blowedge$, $\mu(e)$ is an integral multiple of $2f(\blowvertex)$, therefore it must be $\mu(F) = \mu(F^*)$.  Hence $F$ is a maximum weight perfect $f$-factor of $\blowgraph$.
\end{proof}

\newcommand{\old}{\text{old}}
\newcommand{\walks}{\mathcal{P}}

\section{Running Time Analysis}\label{sec:time}
Next we analyze the running time of the Scaling algorithm~\ref{scaling}. The following lemma constitutes the main technicalities of our analysis. With the assumption that Lemma~\ref{bottleneck} is true, we can finish the proof of Theorem~\ref{nmain}.
\begin{lemma}\label{bottleneck}
	Assume $F_0$ is an arbitrary perfect $f$-factor and let $F_t$ denote the $f$-factor at the end of the $t$-th scaling iteration. For any $t \ge 1$, $F_{t-1}\oplus F_t$ contains at most $\frac{C}{2} n^{2/3}$ edge disjoint augmenting walks in $\blowgraph$, where $C$ is a large constant. 
\end{lemma}
\begin{proof}[Proof of Theorem~\ref{nmain}]
	First let us try to analyze the running time of the $t$-th iteration, where $ t \geq 1$. Clearly the scaling step and the blossom dissolution step only take linear time. By Lemma~\ref{boundedrunningtime}, the deficiency reduction step takes $\tilde{O}(mn^{2/3})$ in total. So the only technical part is the running time of the augmentations within small blossoms.

	When we add edges in $\delta_{F_0}(B_i)\setminus\{\eta(B_i)\}$ to $F$, the over all deficiency of vertices in $B_i$ is thus at most $1+3\binom{|B_i\cap V|}{2} < 1.5n^{2/3}$. By Lemma~\ref{pq-search}, each instance of the \textsf{PQ-Edmonds} algorithm takes $O(m\log n)$ time. After this instance is complete, we claim either (1) the overall deficiency is reduced by one, or (2) the largest $y$ value of unsaturated vertices in $B_i$ is equal to $Y_2$. In fact, if no augmenting walk is found, then on the one hand, the duals of all vertices in $U$ has decreased by $Y_1 - Y_2$, and so their current dual is equal to $Y_2$; all the rest unsaturated vertices have not changed their duals, as the recover stage only modifies duals of matched auxiliary vertices which must be saturated ones. So in this case, the maximum dual value of unsaturated vertices has decreased to $Y_2$.
	
	The former case (1) could happen at most $1.5n^{2/3}$ times, while the latter case (2) could happen at most $|B_i|\leq n^{1/3} + 2\binom{|B_i\cap V|}{2} < n^{2/3}$ times since every time this case happens we add at least one more unsaturated vertex to $U$. In total, the \textsf{PQ-Edmonds} algorithm is invoked for at most $1.5n^{2/3} + n^{2/3} = O(n^{2/3})$ times, and thus the augmentations within small blossoms cost $\tilde{O}(mn^{2/3})$ time.

	Now turn to the PQ-deficiency reduction step. First we claim that the total deficiency of $F_t$ is at most $C n^{2/3}t$ for any $t\geq 1$. By Lemma~\ref{bottleneck}, $F_{1}$ contains at most $\frac{C}{2} n^{2/3}$ edge disjoint augmenting walks in $\blowgraph$ and the result is obviously true. For $t > 1$, assume the total deficiency of $F_{t-1}$ is at most $C n^{2/3}(t-1)$. Since $F_{t-1}\oplus F_t$ contains at most $\frac{C}{2} n^{2/3}$ edge disjoint augmenting walks in $\blowgraph$, the total deficiency of $F_t$ is at most $C n^{2/3}t$. 

	Since the total deficiency by then is at most $\tilde{O}(n^{2/3}\log W)$, we repeatedly apply the PQ-Edmonds algorithm at most $\tilde{O}(n^{2/3}\log W)$ times. By Lemma~\ref{pq-search}, the total running time of the PQ-deficiency reduction step is $\tilde{O}(mn^{2/3}\log W)$.

	Overall, the running time of our scaling algorithm is bounded by $\tilde{O}(mn^{2/3}\log W)$.
\end{proof}

The rest of this section is devoted to the proof of Lemma~\ref{bottleneck}. We prove it by an induction on $t\geq 0$. Suppose right before the $t$-th scaling iteration, we already have computed $F_{t-1}$ as an $f$-factor from the previous scaling iteration. Assume the total deficiency of $F_{t-1}$ is at most $Cn^{2/3}(t-1)$. To find a contradiction, suppose that the total deficiency of $F_t$ at the end of the $t$-th scaling iteration is more than $Cn^{2/3}t$. Then we try to prove the impossibility of $F_{t-1}\oplus F_t$ containing more than $\frac{C}{2}n^{2/3}$ edge disjoint augmenting walks in $\blowgraph$.

One technical issue the induction is when $t =1$, we do not have $F_0$ before the first scaling iteration. Fortunately, we can safely assume $F_0$ is an arbitrary perfect $f$-factor. Note we do not need to explicitly compute $F_0$, but only use it in the analysis.  

For the rest of this section, for convenience, with a slight abuse of notations, define $F_0 = F_{t-1}$ and $F = F_t$. When talking about augmenting and alternating walks, we always mean augmenting walks in $F_0\oplus F$.

\subsection{Some basic tools}
\begin{defn}\label{def-diff}
	Let $\hhat{F}, \hhat{\Omega}, \hhat{z}$ denote any $f$-factor together with a compatible set of blossoms as well as their duals, and let $\rho$ be an arbitrary alternating walk. For any blossom $X\in\hhat{\Omega}$, define the following quantity:
	$$\diff(\rho, X, \hhat{F}) \overset{\mathrm{def}}{=} |\rho\cap \hhat{F}\cap (\gamma(X)\cup I(X))| - |\rho\cap (\gamma(X)\cup I(X))\setminus \hhat{F}|$$
	For any subset of blossoms $S\subseteq \hhat{\Omega}$, define 
	$\diff(\rho, S, \hhat{z}, \hhat{F}) \overset{\mathrm{def}}{=} \sum_{X\in S}\hhat{z}(X)\cdot \diff(\rho, X, \hhat{F})$; basically, $\diff(\rho, S, \hhat{z}, \hhat{F})$ is a weighted summation of $\diff(\rho, X, \hhat{F})$ over all $X\in S$.
\end{defn}

\begin{lemma}\label{diff0}
	Let $\hhat{F}, \hhat{\Omega}$ denote an arbitrary $f$-factor together with a compatible set of blossoms. For any $X\in\hhat{\Omega}$, any alternating walk $\rho$ with length at least $2$. Then, $\diff(\rho, X, \hhat{F})\geq -1$, and if equality holds then either (1) $\eta(X)\in\hhat{F}$ and $\eta(X)\in \rho$, or (2) $\gamma(X)\cup \{\eta(X) \}$ contains $\rho$ entirely.
\end{lemma}
\begin{proof}
	If $\gamma(X)\cup I(X)$ contains $\rho$ entirely, then $\diff(\rho, X, \hhat{F})$ belongs to $\{-1, 0, 1\}$, because $\rho$ alternates between $\hhat{F}\cap (\gamma(X)\cup I(X))$ and $(\gamma(X)\cup I(X))\setminus \hhat{F}$.  When it is $-1$, both of its ending edges are not in $\hhat{F}$, and thus $\gamma(X)\cup \{\eta(X) \}$ contains $\rho$, which corresponds to condition (2).

	For the rest of the proof, suppose $\gamma(X)\cup I(X)$ does not contain $\rho$ entirely. Divide $\rho$ into maximal consecutive sub-walks that lie entirely within $\gamma(X)\cup I(X)$, it suffices to prove lower bounds for each such sub-walk $\rho^\prime$; we do not need to worry about edges that are outside of $\gamma(X)\cup I(X)$ since they do not contribute to $\diff(\rho, X, \hhat{F})$.
	
	If $\rho^\prime$ has even length, then clearly $\diff(\rho^\prime, X, \hhat{F}) = 0$, because $\rho^\prime$ alternates between $\hhat{F}\cap (\gamma(X)\cup I(X))$ and $(\gamma(X)\cup I(X))\setminus \hhat{F}$.  Otherwise, $\diff(\rho^\prime, X, \hhat{F})$ could only be negative when both the starting and ending edges are not in $\hhat{F}$. Let $e_1, e_2$ be the two ending edges of this sub-walk. Clearly, $e_1, e_2$ cannot both be ending edges of $\rho$ as $\rho$ is not contained within $\gamma(X)\cup I(X)$.
	
	Suppose $e_2$ is not $\rho$'s ending edge, then $\rho$ extends $\rho^\prime$ at $e_2$ with $e_3\in \hhat{F}$. By maximality of $\rho^\prime$, $e_3\notin \gamma(X)\cup I(X)$, and thus either $e_2$ or $e_3$ must be equal to $\eta(X)$ and it is a matching edge. Furthermore, a sub-walk can have a negative value of $\diff(\rho^\prime, X, \hhat{F})$, namely $-1$, only when both of its ending edges are not in $\hhat{F}$ and is extended by $\eta(X)$. Such a sub-walk must be unique as $\eta(X)\in\hhat{X}$ cannot extend two different sub-walks within $\gamma(X)\cup I(X)$. To conclude, at most one of the sub-walk $\rho^\prime$ has $\diff(\rho^\prime, X, \hhat{F}) = -1$, while other sub-walks have non-negative values, and therefore $\diff(\rho, X, \hhat{F}) \geq -1$. When equality holds it must be $\eta(X)\in \rho$ and $\eta(X)$ is a matching edge, which corresponds to condition (1).
\end{proof}

We also need a sufficient condition to ensure $\diff(\rho, X, \hhat{F})\geq 0$.

\begin{lemma}\label{diff}
	Let $\hhat{F}, \hhat{\Omega}$ denote some $f$-factor together with a compatible set of blossoms. For any $X\in\hhat{\Omega}$, any alternating walk $\rho$ with length at least $2$, if both two ending edges are not from $\gamma(X)\cup I(X)\setminus \hhat{F}$, then $\diff(\rho, X, \hhat{F})\geq 0$.
\end{lemma}
\begin{proof}
	Similar to the previous proof, divide $\rho$ into maximal consecutive sub-walks that lie entirely within $\gamma(X)\cup I(X)$, it suffices to for each such sub-walk $\rho^\prime$ that $\diff(\rho\prime, X, \hhat{F})\geq 0$. In fact, if $\rho^\prime$ has even length, then clearly $\diff(\rho^\prime, X, \hhat{F}) = 0$. Otherwise, $\diff(\rho^\prime, X, \hhat{F})$ could only be negative when both the starting and ending edges are not in $\hhat{F}$. Let $e_1, e_2$ be the two ending edges of this sub-walk.

	Since both $e_1, e_2$ cannot be $\rho$'s ending edges, $\rho$ must extend $\rho^\prime$ from both of them. Assume $\rho$ extends $\rho^\prime$ at $e_1$ with an edge $e_0\in \hhat{F}$, and at $e_2$ with an edge $e_3\in \hhat{F}$. By maximality of $\rho^\prime$, $e_i\notin \gamma(X)\cup I(X), i\in \{0, 3\}$. Consider the following cases.
	\begin{itemize}
		\item $\eta(X)\neq\hhat{F}$. Suppose $\eta(X) = (u, \beta(X))$, then the only possibility that $e_0\notin \gamma(X)\cup I(X)$ is $e_1 = \eta(X)$ and $e_0\cap e_1 = \{u \}$. By symmetry $e_2 = \eta(X)$ and $e_2\cap e_3 = \{u\}$. Hence, $\rho^\prime = \eta(X)$ is a single edge, and $\rho$ extends $\rho^\prime$ only from one endpoint, namely $u$, of $\eta(X)$, which is impossible.
		\item $\eta(X) \in \hhat{F}$. Suppose $\eta(X) = (u, \beta(X))$. In this case, the only possibility that $e_0\notin \gamma(X)\cup I(X)$ is $e_0 = \eta(X)$ and $e_0\cap e_1 = \{\beta(X) \}$. By symmetry $e_3 = \eta(X)$ and $e_2\cap e_3 = \{\beta(X) \}$. Hence, $\rho$ repeats itself at $e_0 = e_3 = \eta(X)$ which is impossible.
		
		\item $\eta(X) = \NULL$. In this case, $e_0, e_3\in \hhat{X}$ always belong to $\gamma(X)\cup I(X)$, which contradicts maximality of $\rho^\prime$.
	\end{itemize}
\end{proof}

Let $\scale_\old, y_\old, z_\old, \Omega_\old$ denote the edge weights, duals, and blossoms at the beginning of the blossom dissolution step, and let $\Omega_\old^\text{large}$ be the set of all blossoms in $\Omega_\old$ that were dissolved in the blossom dissolution phase before the reweighting step. Let $\Omega^\text{large}$ be the set of all large blossoms in $\Omega$.

\begin{lemma}\label{unique}
	For any blossom $X\in \Omega\cup \Omega_\old$, there exists at most one augmenting walk $\rho$ from $F_0\oplus F$ such that $\rho\cap \{\eta(X), \zeta(X) \}\neq \emptyset$.
\end{lemma}
\begin{proof}
	Suppose otherwise there are two different augmenting walks $\rho_1, \rho_2$ that intersects $\{\eta(X),\zeta(X) \}$, then $\eta(X), \zeta(X)\neq \NULL$. Suppose $\eta(X) = (u, e_u)$ and $\zeta(X) = (e_u, e_v)$. Since $\rho_1, \rho_2$ are edge disjoint and they contain $\eta(X)$ and $\zeta(X)$ respectively, $e_u$ must be a common ending point of both augmenting walks. Hence $(u, e_u), (e_u, e_v)\in F_0$, which is impossible since $f(e_u) = 1$.
\end{proof}

\begin{lemma}\label{z-bound}
	$\sum_{B\in\Omega^\text{large}}z(B)\leq 2Cn^{4/3} +14n^{2/3}$ for a large enough constant $C$.
\end{lemma}
\begin{proof}
	This is because, the total sum of duals of large blossom increases by at most $2n^{2/3}$ after each dual adjustment in the deficiency reduction phase, and thus the total sum at the end of the current scaling iteration is at most $2n^{2/3}\cdot (\ceil{Cn^{2/3}}+6)<2Cn^{4/3} + 14n^{2/3}$.
\end{proof}

\begin{lemma}\label{z1-bound}
	$\sum_{B\in\Omega_\old^\text{large}}z_\old(B)\leq 5Cn^{4/3}$ for a large enough constant $C$.
\end{lemma}
\begin{proof}
	We first argue that $\sum_{\text{large }B\in\Omega_\old}z_\old(B)\leq 4Cn^{4/3}$; recall that $\Omega_\old^\text{large}$ not only contains all large blossoms in $\Omega_\old$, but some small blossoms with small dual values as well. This is because, in the previous iteration, by Lemma~\ref{z-bound} the total sum at the end of the previous scaling iteration is at most $2Cn^{4/3} + 14n^{2/3}$. Therefore, after the scaling phase, the total sum of $z_\old(\cdot)$ of large blossoms in $\Omega_\old$ is at most $4Cn^{4/3} + 28n^{2/3}$. On the other hand, by definition of $\Omega_\old^\text{large}$, other than large blossoms, $\Omega_\old^\text{large}$ also contains small blossom $B$ such that $z_\old(B)\leq \gap$, and thus the total sum of duals of these blossoms is at most $12n$. Finally, $\sum_{B\in\Omega_\old^\text{large}}z_\old(B)\leq 4Cn^{4/3} + 28n^{2/3} + 12n \leq 5Cn^{4/3}$ for a large enough constant $C$.
\end{proof}

\subsection{The main proof}

\textbf{General strategy}: Let $\walks$ be a set of augmenting walks which is initialized to the $F_0\oplus F$. Assume $|\walks| \geq \frac{C}{2}n^{2/3}$. To reach a contradiction, our analysis of $\walks$ will consist of several phases; in each phase the set $|\walks|$ is pruned according to some criteria but we will still be guaranteed a lower bound on $|\walks|$.

\subsection*{Phase 1}

Instead of directly working with duals $y$, define variables as follows:
$$\dy(u) = y(u) + \frac{1}{2}\sum_{\substack{X\in\Omega^\text{large} \\ \exists v, (u, v)\in \gamma(X)\cup I(X)}}z(X)$$

\begin{defn}
	Call an unsaturated vertex $u$ \emph{unaffected} if $\dy(u) = y(u)$. An augmenting walk is \emph{unaffected} if both of its unsaturated endpoints are unaffected.
\end{defn}

The next lemma claims that most of the unsaturated vertices are unaffected.

\begin{lemma}\label{unsat}
	There are at least $(C-1)n^{2/3}$ unaffected vertices.
\end{lemma}
\begin{proof}
	We claim every large blossom can affect at most one unsaturated vertex. For each $X\in\Omega^\text{large}$, $z(X)$ could contribute to at most one unsaturated vertex $u\in X$, since $X$ contains at most one unsaturated vertex. Next we only need to worry about vertices $u\notin X$ that are affected by $X$. Then, there exists $v\in X$ such that $(u, v)\in I(X)$. Since $v$ belongs to $X$ which is a nontrivial blossom, $v$ must be an original vertex in $V$, and hence $u$ is an auxiliary vertex. Consider two cases.
	
	\begin{itemize}
		\item $(u, v) \in F$. In this case, since $u$ is an auxiliary vertex, $f(u) = 1$. As $(u, v)\in F$, $u$ is already saturated.
		
		\item $(u, v)\notin F$. In this case, as $(u, v)$ belongs to $I(X)$, it must be $(u, v) = \eta(X)$ and $\eta(X)\notin F$. Hence such $u$  which equals to an endpoint  of $\eta(X)$ is unique.
	\end{itemize}
	
	To summarize, $X$ can affect at most one unsaturated vertex outside of $X$; when this happens, it should be $\eta(X)\neq \NULL$, and thus $X$ does not affect any unsaturated vertex within $X$ since all of them are saturated. As there are at most $n^{2/3}$ root blossoms in $\Omega^\text{large}$, there are at most $n^{2/3}$ affected unsaturated vertices, and therefore there are at least $(C-1)n^{2/3}$ unaffected unsaturated vertices.
\end{proof}
\begin{corollary}
	There are at least $(\frac{1}{2}C-1)n^{2/3}$ unaffected augmenting walks.
\end{corollary}

For the rest of this section we will only be looking at unaffected augmenting walks. Namely, remove all affected augmenting walks from $\walks$. By the above corollary, we still have $|\walks|\geq (\frac{1}{2}C - 1)n^{2/3}$.

\subsection*{Phase 2}

\begin{lemma}\label{odd}
	Consider any alternating walk $\rho = \langle u_1, u_2, \cdots, u_{2s+1}\rangle$ starting with an edge not in $F$. Then, 
	$$\begin{aligned}
	\dy(u_{2s+1}) - \dy(u_1)&\leq 8s + \frac{1}{2}\sum_{i=1}^s\sum_{\substack{X\in\Omega_\old^\text{large} \\ (u_{2i-1}, u_{2i}) \in\{ \eta(X), \zeta(X)\}}}z_\old(X) + \frac{1}{2}\sum_{i=1}^s\sum_{\substack{X\in\Omega^\text{large} \\ (u_{2i}, u_{2i+1}) \in \{\eta(X),\zeta(X) \}}}z(X)\\
	&- \diff(\rho, \Omega\setminus \Omega^\text{large}, z, F)-\diff(\rho, \Omega_\old\setminus\Omega_\old^\text{large}, z_\old, F_0)
	\end{aligned}$$
\end{lemma}

This lemma tries to analyze the difference between $\dy(u_{2s+1})$ and $\dy(u_1)$. Speaking on a high level, the two summation terms in the middle $$\frac{1}{2}\sum_{i=1}^s\sum_{\substack{X\in\Omega_\old^\text{large} \\ (u_{2i-1}, u_{2i}) \in\{\eta(X),\zeta(X) \}}}z_\old(X) + \frac{1}{2}\sum_{i=1}^s\sum_{\substack{X\in\Omega^\text{large} \\ (u_{2i}, u_{2i+1}) \in\{\eta(X),\zeta(X) \} }}z(X)$$ will be small on average since the total sum of duals of large blossoms are bounded by $O(Cn^{4/3})$. For the last two terms, $\diff(\rho, \Omega\setminus \Omega^\text{large}, z, F)+\diff(\rho, \Omega_\old\setminus\Omega_\old^\text{large}, z_\old, F_0)$, their influences are also very limited as the total number of edges in small blossoms is bounded by $O(n^{4/3})$. Therefore, when $s$ is large, $\dy(u_{2s+1}) - \dy(u_1)$ roughly grows linearly with $s$.

\begin{proof}[Proof of Lemma~\ref{odd}]
	Consider any index $1\leq i\leq s$ and study $\dy(u_{2i+1}) - \dy(u_{2i-1})$. By the dominance condition from approximate complementary slackness we have 
	$$y(u_{2i-1}) + y(u_{2i}) + \sum_{\substack{X\in \Omega \\ (u_{2i-1}, u_{2i})\in \gamma(X)\cup I(X)}}z(X)\geq \scale(u_{2i-1}, u_{2i}) - 2$$
	
	Plugging in the definition of $\dy$, it leads to
	$$\dy(u_{2i-1}) + \dy(u_{2i}) + \sum_{\substack{X\in\Omega\setminus\Omega^\text{large} \\ (u_{2i-1}, u_{2i})\in\gamma(X)\cup I(X)}}z(X) \geq \scale(u_{2i-1}, u_{2i})-2$$
	
	By the tightness condition and $(u_{2i}, u_{2i+1})\in F$, 
	$$y(u_{2i}) + y(u_{2i+1}) + \sum_{\substack{X\in\Omega \\ (u_{2i}, u_{2i+1})\in \gamma(X)\cup I(X)}}z(X)\leq \scale(u_{2i}, u_{2i+1})$$
	
	Plugging in the definition of $\dy$ and $(u_{2i}, u_{2i+1})\in F$, we always have
	$$\dy(u_{2i}) + \dy(u_{2i+1}) = y(u_{2i}) + y(u_{2i+1}) + \sum_{\substack{X\in\Omega^\text{large} \\ (u_{2i}, u_{2i+1})\in\gamma(X)\cup I(X)}}z(X) + \frac{1}{2}\sum_{\substack{X\in\Omega^\text{large} \\ (u_{2i}, u_{2i+1}) \in\{\eta(X),\zeta(X) \}}}z(X)$$ and therefore
	$$\dy(u_{2i}) + \dy(u_{2i+1}) + \sum_{\substack{X\in\Omega\setminus\Omega^\text{large} \\ (u_{2i}, u_{2i+1})\in\gamma(X)\cup I(X)}}z(X)\leq \scale(u_{2i}, u_{2i+1}) + \frac{1}{2}\sum_{\substack{X\in\Omega^\text{large} \\ (u_{2i}, u_{2i+1}) \in\{\eta(X),\zeta(X) \}}}z(X)$$
	
	Taking a subtraction we have
	$$\begin{aligned}
	\dy(u_{2i+1}) - \dy(u_{2i-1})&\leq 2+\scale(u_{2i}, u_{2i+1}) - \scale(u_{2i-1}, u_{2i}) +\frac{1}{2}\sum_{\substack{X\in\Omega^\text{large}\\ (u_{2i}, u_{2i+1}) \in\{\eta(X),\zeta(X) \}}}z(X)\\
	&+ \sum_{\substack{X\in \Omega\setminus\Omega^\text{large}\\ (u_{2i-1}, u_{2i})\in \gamma(X)\cup I(X)}}z(X) - \sum_{\substack{X\in\Omega\setminus\Omega^\text{large} \\ (u_{2i}, u_{2i+1})\in \gamma(X)\cup I(X)}}z(X)
	\end{aligned}$$
	
	By a summation over all $1\leq i\leq s$, it follows
	$$\begin{aligned}
	\dy(u_{2s+1}) - \dy(u_1)&\leq 2s + \sum_{i = 1}^s \left(\scale(u_{2i}, u_{2i+1})-\scale(u_{2i-1}, u_{2i}) \right)+ \frac{1}{2}\sum_{i=1}^s\sum_{\substack{X\in\Omega^\text{large} \\ (u_{2i}, u_{2i+1}) \in\{\eta(X),\zeta(X) \}}}z(X)\\
	&+ \sum_{i=1}^s\left(\sum_{\substack{X\in \Omega\setminus\Omega^\text{large} \\ (u_{2i-1}, u_{2i})\in \gamma(X)\cup I(X)}}z(X) - \sum_{\substack{X\in\Omega\setminus\Omega^\text{large} \\ (u_{2i}, u_{2i+1})\in \gamma(X)\cup I(X)}}z(X) \right)\\
	&= 2s + \frac{1}{2}\sum_{i=1}^s\sum_{\substack{X\in\Omega^\text{large} \\ (u_{2i}, u_{2i+1}) \in\{\eta(X),\zeta(X) \}}}z(X)\\
	&+\sum_{i = 1}^s (\scale(u_{2i}, u_{2i+1})-\scale(u_{2i-1}, u_{2i})) - \diff(\rho, \Omega\setminus \Omega^\text{large}, z, F)
	\end{aligned}$$
	
	Next we set out to analyze the first summation $\sum_{i = 1}^s (\scale(u_{2i}, u_{2i+1})-\scale(u_{2i-1}, u_{2i}))$. As with Lemma~\ref{half}, recall $\scale_\old, y_\old, z_\old, \Omega_\old$ denote the edge weights, duals, and blossoms at the beginning of the blossom dissolution step, and let $\Omega_\old^\text{large}$ be the set of all blossoms in $\Omega_\old$ that were dissolved in the blossom dissolution step in line 7-8. Then we have the following:
	$$\begin{aligned}
	\scale(u_{2i}, u_{2i+1}) 
	&\leq yz_\old(u_{2i}, u_{2i+1}) - y_\old(u_{2i}) - y_\old(u_{2i+1}) -\sum_{\substack{X\in\Omega_\old^\text{large} \\ (u_{2i}, u_{2i+1})\in\gamma(X)\cup I(X)}}z_\old(X)\\
	&= \sum_{\substack{X\in\Omega_\old\setminus\Omega_\old^\text{large} \\ (u_{2i}, u_{2i+1})\in\gamma(X)\cup I(X)}}z_\old(X)
	\end{aligned}$$
	The first inequality is by Lemma~\ref{scale-slack}. 
	
	Now, as $(u_{2i-1}, u_{2i})\in F_0$, again by Lemma~\ref{scale-slack}, 
	$$\begin{aligned}
	\scale(u_{2i-1}, u_{2i}) &= \scale_\old(u_{2i-1}, u_{2i}) - y_\old(u_{2i-1}) - y_\old(u_{2i})\\
	&- \sum_{\substack{X\in\Omega_\old^\text{large} \\ (u_{2i-1}, u_{2i})\in\gamma(X)\cup I(X)}}z_\old(X) - \frac{1}{2}\sum_{\substack{X\in\Omega_\old^\text{large} \\ (u_{2i-1}, u_{2i})\in\{\eta(X),\zeta(X) \}}}z_\old(X)\\
	&\geq yz_\old(u_{2i-1}, u_{2i}) -6- y_\old(u_{2i-1}) - y_\old(u_{2i})\\
	&- \sum_{\substack{X\in\Omega_\old^\text{large} \\ (u_{2i-1}, u_{2i})\in\gamma(X)\cup I(X)}}z_\old(X)-\frac{1}{2}\sum_{\substack{X\in\Omega_\old^\text{large} \\ (u_{2i-1}, u_{2i})\in\{\eta(X),\zeta(X) \}}}z_\old(X)\\
	&= -6 +  \sum_{\substack{X\in\Omega_\old\setminus\Omega_\old^\text{large} \\ (u_{2i-1}, u_{2i})\in\gamma(X)\cup I(X)}}z_\old(X)-\frac{1}{2}\sum_{\substack{X\in\Omega_\old^\text{large} \\ (u_{2i-1}, u_{2i})\in\{\eta(X),\zeta(X) \} }}z_\old(X)
	\end{aligned}$$
	
	Taking a summation we have,
	$$\begin{aligned}
	&\sum_{i = 1}^s \left(\scale(u_{2i}, u_{2i+1})-\scale(u_{2i-1}, u_{2i})\right)\leq 6s + \frac{1}{2}\sum_{i=1}^s\sum_{\substack{X\in\Omega_\old^\text{large} \\ (u_{2i-1}, u_{2i}) \in\{\eta(X),\zeta(X) \}}}z_\old(X)\\
	&+\sum_{i=1}^s\left(\sum_{\substack{X\in\Omega_\old\setminus\Omega_\old^\text{large} \\ (u_{2i}, u_{2i+1})\in\gamma(X)\cup I(X)}}z_\old(X)-\sum_{\substack{X\in\Omega_\old\setminus\Omega_\old^\text{large} \\ (u_{2i-1}, u_{2i})\in\gamma(X)\cup I(X)}}z_\old(X)\right)\\
	&= 6s + \frac{1}{2}\sum_{i=1}^s\sum_{\substack{X\in\Omega_\old^\text{large} \\ (u_{2i-1}, u_{2i}) \in\{\eta(X),\zeta(X) \}}}z_\old(X) - \diff(\rho, \Omega_\old\setminus\Omega_\old^\text{large}, z_\old, F_0)
	\end{aligned}$$ which concludes the proof.	
\end{proof}

Besides the relation between $\dy(u_{2s+1})$ and $\dy(u_1)$, we also need some relations between $\dy(u_{2s})$ and $\dy(u_1)$, which is stated as two lemmas coming next.

\begin{lemma}\label{even}
	Consider any alternating walk $\rho = \langle u_1, u_2, \cdots, u_{2s}\rangle$ starting with an edge not in $F$. Then, 
	$$\begin{aligned}\dy(u_{2s})&\geq -\dy(u_1)-8s-2 -\frac{1}{2}\sum_{i=1}^s\sum_{\substack{X\in\Omega_\old^\text{large} \\ (u_{2i-1}, u_{2i})\in\{\eta(X),\zeta(X) \}}}z_\old(X) -\frac{1}{2}\sum_{i=1}^{s-1}\sum_{\substack{X\in\Omega^\text{large} \\ (u_{2i}, u_{2i+1})\in\{\eta(X),\zeta(X) \}}}z(X) \\
	&+\diff(\rho, \Omega_\old\setminus \Omega_\old^\text{large}, z_\old, F_0)
	+\diff(\rho, \Omega\setminus \Omega^\text{large}, z, F)
	\end{aligned}$$
\end{lemma}
\begin{proof}
	Similar to the derivation in the previous lemma, we have:
	$$\dy(u_{2s-1}) + \dy(u_{2s}) + \sum_{\substack{X\in\Omega\setminus\Omega^\text{large} \\ (u_{2s-1}, u_{2s})\in\gamma(X)\cup I(X)}}z(X) \geq \scale(u_{2s-1}, u_{2s})-2$$

	$$\scale(u_{2s-1}, u_{2s})\geq -6 +  \sum_{\substack{X\in\Omega_\old\setminus\Omega_\old^\text{large} \\ (u_{2s-1}, u_{2s})\in\gamma(X)\cup I(X)}}z_\old(X)-\frac{1}{2}\sum_{\substack{X\in\Omega_\old^\text{large} \\ (u_{2s-1}, u_{2s})\in\{\eta(X),\zeta(X) \}}}z_\old(X)$$
	
	Combining the two inequalities and by Definition~\ref{def-diff}, 
	$$\begin{aligned}
	\dy(u_{2s})&\geq -\dy(u_{2s-1})-8 -\frac{1}{2}\sum_{\substack{X\in\Omega_\old^\text{large} \\ (u_{2s-1}, u_{2s})\in\{\eta(X),\zeta(X) \}}}z_\old(X)\\
	&+\diff((u_{2s-1}, u_{2s}), \Omega_\old\setminus \Omega_\old^\text{large}, z_\old, F_0)
	+\diff((u_{2s-1}, u_{2s}), \Omega\setminus\Omega^\text{large}, z, F)
	\end{aligned}$$
	
	Plugging Lemma~\ref{odd}, we have
	$$\begin{aligned}\dy(u_{2s})&\geq -\dy(u_1)-8s-2 -\frac{1}{2}\sum_{i=1}^s\sum_{\substack{X\in\Omega_\old^\text{large} \\ (u_{2i-1}, u_{2i})\in\{\eta(X),\zeta(X) \}}}z_\old(X) -\frac{1}{2}\sum_{i=1}^{s-1}\sum_{\substack{X\in\Omega^\text{large} \\ (u_{2i}, u_{2i+1})\in\{\eta(X),\zeta(X) \}}}z(X) \\
	&+\diff(\rho, \Omega_\old\setminus \Omega_\old^\text{large}, z_\old, F_0)
	+\diff(\rho, \Omega\setminus \Omega^\text{large}, z, F)
	\end{aligned}$$
\end{proof}

\begin{lemma}\label{even1}
	Consider any alternating walk $\rho = \langle u_1, u_2, \cdots, u_{2s}\rangle$ starting with an edge in $F$. Then, 
	$$\begin{aligned}\dy(u_{2s})&\leq -\dy(u_1)+8s +\frac{1}{2}\sum_{i=1}^{s-1}\sum_{\substack{X\in\Omega_\old^\text{large} \\ (u_{2i}, u_{2i+1}) \in\{ \eta(X), \zeta(X)\} }}z_\old(X) + \frac{1}{2}\sum_{i=1}^s\sum_{\substack{X\in\Omega^\text{large}\\ (u_{2i-1}, u_{2i}) \in\{ \eta(X), \zeta(X)\}}}z(X) \\
	&-\diff(\rho, \Omega_\old\setminus \Omega_\old^\text{large}, z_\old, F_0)
	-\diff(\rho, \Omega\setminus \Omega^\text{large}, z, F)
	\end{aligned}$$
\end{lemma}

To avoid possible confusions, we emphasize the different between this lemma and Lemma~\ref{even}: here the alternating walk $\rho$ starts with a matching edge from $F$ while in the previous lemma the alternating walk $\rho$ starts with a non-matching edge from $\blowedge\setminus F$.

\begin{proof}[Proof of Lemma~\ref{even1}]
	Since $(u_{2s-1}, u_{2s})\in F$, borrowing the derivation from the Lemma~\ref{odd}, 
	$$\dy(u_{2s-1}) + \dy(u_{2s}) + \sum_{\substack{X\in\Omega\setminus\Omega^\text{large} \\ (u_{2s-1}, u_{2s})\in\gamma(X)\cup I(X)}}z(X) \leq \scale(u_{2s-1}, u_{2s}) + \frac{1}{2}\sum_{\substack{X\in\Omega^\text{large} \\ (u_{2s-1}, u_{2s}) \in\{\eta(X),\zeta(X) \}}}z(X)$$
	
	$$\scale(u_{2s-1}, u_{2s})\leq  \sum_{\substack{X\in\Omega_\old\setminus\Omega_\old^\text{large} \\ (u_{2s-1}, u_{2s})\in\gamma(X)\cup I(X)}}z_\old(X)$$
	
	Combining the two inequalities and by Definition~\ref{def-diff}, 
	$$\begin{aligned}
	\dy(u_{2s})&\leq -\dy(u_{2s-1}) +\frac{1}{2}\sum_{\substack{X\in\Omega^\text{large} \\ (u_{2s-1}, u_{2s})\in\{\eta(X),\zeta(X) \}}}z(X)\\
	&-\diff((u_{2s-1}, u_{2s}), \Omega_\old\setminus \Omega_\old^\text{large}, z_\old, F_0)
	-\diff((u_{2s-1}, u_{2s}), \Omega\setminus\Omega^\text{large}, z, F)
	\end{aligned}$$
	
	Plugging Lemma~\ref{odd} on the alternating walk $\langle u_{2s-1}, u_{2s-2}, \cdots, u_1\rangle$ which starts with non-matching edge $(u_{2s-1}, u_{2s-2})\notin F$, we have
	$$\begin{aligned}\dy(u_{2s})&\leq -\dy(u_1)+8s +\frac{1}{2}\sum_{i=1}^{s-1}\sum_{\substack{X\in\Omega_\old^\text{large} \\ (u_{2i}, u_{2i+1}) \in\{\eta(X),\zeta(X) \}}}z_\old(X) + \frac{1}{2}\sum_{i=1}^s\sum_{\substack{X\in\Omega^\text{large} \\ (u_{2i-1}, u_{2i})\in\{\eta(X),\zeta(X) \}}}z(X) \\
	&-\diff(\rho, \Omega_\old\setminus \Omega_\old^\text{large}, z_\old, F_0)
	-\diff(\rho, \Omega\setminus \Omega^\text{large}, z, F)
	\end{aligned}$$
\end{proof}

What we do next is to argue that most of the augmenting walk in $F\oplus F_0$ has (unweighted) length $\Omega(Cn^{2/3})$.

\begin{lemma}\label{length}
	Consider any augmenting walk $\rho = \langle u_1, u_2, \cdots, u_{2s}\rangle$ such that $u_1, u_{2s}$ are both unaffected. Then, 
	$$s \geq \frac{1}{4}Cn^{2/3} - \frac{1}{4} - \frac{1}{16}\sum_{i=1}^s\sum_{\substack{X\in\Omega_\old^\text{large} \\ (u_{2i-1}, u_{2i})\in\{\eta(X),\zeta(X) \}}}z_\old(X) - \frac{1}{16}\sum_{i=1}^{s-1}\sum_{\substack{X\in\Omega^\text{large} \\ (u_{2i}, u_{2i+1})\in\{\eta(X),\zeta(X) \}}}z(X)$$
\end{lemma}
\begin{proof}
	Plugging in Lemma~\ref{even} and by $\dy(u_{2s}) = \dy(u_1) = y(u_1) \leq -Cn^{2/3}$, we have
	$$\begin{aligned}s&\geq \frac{1}{4}Cn^{2/3} - \frac{1}{4} - \frac{1}{16}\sum_{i=1}^s\sum_{\substack{X\in\Omega_\old^\text{large}\\ (u_{2i-1}, u_{2i})\in\{\eta(X),\zeta(X) \}}}z_\old(X) - \frac{1}{16}\sum_{i=1}^{s-1}\sum_{\substack{X\in\Omega^\text{large} \\ (u_{2i}, u_{2i+1})\in\{\eta(X),\zeta(X) \}}}z(X)\\
	&+ \frac{1}{8}\diff(\rho, \Omega_\old\setminus \Omega_\old^\text{large}, z_\old, F_0)
	+\frac{1}{8}\diff(\rho, \Omega\setminus\Omega^\text{large}, z, F)\end{aligned}$$ Then it suffices to prove both $\diff(\rho, \Omega_\old\setminus \Omega_\old^\text{large}, z_\old, F_0)$ and $\diff(\rho, \Omega, z, F)$ are non-negative. For the first term, for any $X\in\Omega_\old\setminus\Omega_\old^{\text{large}}$, since both ending edges of $\rho$ belong to $F_0$ thus not contained in $\gamma(X)\cup I(X)\setminus F_0$, by Lemma~\ref{diff}, $\diff(\rho, X, F_0)\geq 0$. 
	
	As for the second term, consider any blossom $X\in\Omega\setminus \Omega^{\text{large}}$. If $\gamma(X)\cup I(X)$ does not contain any of the two ending edges $(u_1, u_2), (u_{2s-1}, u_{2s})$, then using Lemma~\ref{diff} we know $\diff(\rho, X, F)\geq 0$. Otherwise, assume it contains $(u_1, u_2)$. If $u_1\in X$, then because deficiency of $u_1$ is $1$, we know $\eta(X) = \NULL$, and thus by Lemma~\ref{diff0}, the only possibility for $\diff(\rho, X, F)=-1$ is $\gamma(X)\cup \{\eta(X) \} = \gamma(X)$ contains $\rho$ entirely; this is impossible since any blossom $X$ cannot contain two different unsaturated vertices $u_1, u_{2s}$.
	
	Now consider the case where $u_1\notin X$ but $(u_1, u_2) \in I(X)$. In this case, it must be $(u_1, u_2) = \eta(X)$ as $(u_1, u_2)$ is not matched by $F$, and hence $\eta(X)\notin F$. By Lemma~\ref{diff0}, the only possibility left for $\diff(\rho, X, F) = -1$ is $\gamma(X)\cup \{\eta(X) \}$ contains the entire augmenting walk $\rho$. So in particular, $u_{2s}$ belongs to $X\cup \{u_1 \}$. As $u_{2s}\neq u_1$, it must be $u_{2s}\in X$. However, this is again not possible, since any $X$ containing an unsaturated vertex must satisfy $\eta(X) = \NULL$. To conclude, we can still claim $\diff(\rho, X, F)\geq 0$.
\end{proof}

\begin{corollary}\label{long}
	There are at least $(\frac{1}{4}C-1)n^{2/3}$ augmenting walks in $\walks$ whose length is at least $\frac{1}{3}Cn^{2/3}$.
\end{corollary}
\begin{proof}
	By Lemma~\ref{unique}, Lemma~\ref{z-bound} and Lemma~\ref{z1-bound}
	$$\begin{aligned}
	&\sum_{\rho\in\walks}\sum_{i = 1}^{s-1}\sum_{\substack{X\in\Omega^\text{large} \\ (u_{2i}, u_{2i+1})\in\{\eta(X),\zeta(X) \}}}z(X)\leq 2Cn^{4/3} + 14n^{2/3}\\
	&\sum_{\rho\in\walks}\sum_{i = 1}^s\sum_{\substack{X\in\Omega_\old^\text{large} \\ (u_{2i-1}, u_{2i})\in\{\eta(X),\zeta(X) \}}}z_\old(X)\leq 5Cn^{4/3}\end{aligned}$$
	
	Therefore, for at most $\frac{1}{4}Cn^{2/3}$ unaffected augmenting walks $\rho = \langle u_1, u_2, \cdots, u_{2s}\rangle$, the following summation could be larger:
	
	$$\begin{aligned}
	&\sum_{i=1}^s\sum_{\substack{X\in\Omega_\old^\text{large} \\ (u_{2i-1}, u_{2i})\in\{\eta(X),\zeta(X) \}}}z_\old(X) + \sum_{i=1}^{s-1}\sum_{\substack{X\in\Omega^\text{large} \\ (u_{2i}, u_{2i+1})\in\{\eta(X),\zeta(X) \}}}z(X)\\
	&\geq \frac{2Cn^{4/3} + 14n^{2/3} + 5Cn^{4/3}}{\frac{1}{4}Cn^{2/3}} = 28n^{2/3} + 20\end{aligned}$$
	
	For the rest of $(\frac{1}{4}C-1)n^{2/3}$ unaffected augmenting walks, by Lemma~\ref{length}, the length of this augmenting walk $2s-1\geq \frac{1}{2}Cn^{2/3} - \frac{3}{2} - \frac{1}{16}(28n^{2/3} + 20) = (\frac{1}{2}C - \frac{7}{4})n^{2/3} - \frac{11}{4}\geq \frac{1}{3}Cn^{2/3}$ for a large enough constant $C$.
\end{proof}

Here we do a second round of pruning by this lemma; i.e., remove from $\walks$ all augmenting walks whose length is less than $\frac{1}{3}Cn^{2/3}$. So $|\walks|\geq (\frac{1}{4}C-1)n^{2/3}$.

\subsection*{Phase 3}

Next we try to upper bound the influence of edges that are from small blossoms. 
\begin{defn}
	For an augmenting walk $\rho = \langle u_1, u_2, \cdots, u_{2s}\rangle$, an edge $(u_i, u_{i+1}) \in \rho$ is called \textbf{bad}, if it belongs to some $\gamma(X)\cup I(X)\setminus F, X\in\Omega\setminus\Omega^\text{large}$, or it belongs to some $\gamma(X)\cup I(X)\setminus F_0, X\in\Omega_\old\setminus\Omega_\old^\text{large}$. Call a vertex $u_i$ \textbf{good} if it is not incident on any bad edges on this walk.
\end{defn}

\begin{lemma}
	In the blowup graph $\blowgraph$, the total number of bad edges is at most $3n^{4/3}$.
\end{lemma}
\begin{proof}
	Clearly, bad edges are always in $\gamma(X)\cup\{\eta(X) \}$ as $\gamma(X)\cup I(X)\setminus F\subseteq \gamma(X)\cup \{\eta(X) \}$, for some $X\in(\Omega\setminus \Omega^\text{large})\cup(\Omega_\old\setminus \Omega_\old^\text{large})$. So we only need to bound the total number of edges in $\gamma(X)\cup\{\eta(X) \}$. 
	
	First we claim that the total number of edges in $\gamma(X)\cup \{\eta(X) \}$ connecting two auxiliary vertices $e_u, e_v$ is at most half of the number of edges in $\gamma(X)\cup \{\eta(X) \}$ incident on original vertices. This is because, for any $(e_u, e_v)\in \gamma(X)\cup \{\eta(X) \}$, by the structure of the blowup graph it must be $u, v\in X$, and hence $(u, e_u), (v, e_v)\in \gamma(X)$. Thus we only need to upper bound the number of bad edges of the latter form by a total amount of $2n^{4/3}$.
	
	For each original vertex $u\in V$ in graph $\blowgraph$, let $X$ be the maximal blossom from $\Omega\setminus\Omega^\text{large}$ or $\Omega_\old\setminus\Omega_\old^\text{large}$ that contains $u$. Its degree in $\blowgraph[X]$ is at most $n^{1/3}-1$ since $X$ is always a small blossom. Therefore, including $\eta(X)$, there are at most $n^{1/3}$ bad edges incident on $u$. As there are at most two different maximal $X$'s that contain $u$, one from $\Omega\setminus\Omega^\text{large}$ and one from $\Omega_\old\setminus \Omega_\old^\text{large}$, the total number of bad edges incident on $u$ is at most $2n^{1/3}$. Ranging over all different $u\in V$ finishes the proof.
\end{proof}

\begin{corollary}\label{not-so-bad}
	There are at most $3n^{2/3}$ augmenting walks contain more than $n^{2/3}$ bad edges.
\end{corollary}

For the rest, we are only interested in unaffected augmenting walks containing at most $n^{2/3}$ bad edges. By Corollary~\ref{long} and Corollary~\ref{not-so-bad}, there are at least $(\frac{1}{4}C-1)n^{2/3} - 3n^{2/3} > (\frac{1}{4}C-4)Cn^{2/3}$ many of them. Remove from $\walks$ all augmenting walks that do not satisfy this property. So far we have $|\walks|\geq (\frac{1}{4}C-1)Cn^{2/3}$.

\subsection*{Phase 4}

Next we argue that most of these augmenting walks contain a vertex whose $y$ dual is non-negative. More specifically, take an arbitrary augmenting walk $\rho = \langle u_1, u_2, \cdots, u_{2s}\rangle$. Find the smallest $l$ such that $u_{2l}$ is a good vertex. Since this walk contains at most $n^{2/3}$ bad edges, we are ensured that $l\leq 2n^{2/3}$. Let $\rho_l = \langle u_1, u_2, \cdots, u_{2l}\rangle$ be the prefix alternating walk, then by Lemma~\ref{even}, 
$$\begin{aligned}
\dy(u_{2l})
& \geq Cn^{2/3}-8l-2 -\frac{1}{2}\sum_{i=1}^l\sum_{\substack{X\in\Omega_\old^\text{large}\\ (u_{2i-1}, u_{2i})\in\{\eta(X),\zeta(X) \}}}z_\old(X) -\frac{1}{2}\sum_{i=1}^{l-1}\sum_{\substack{X\in\Omega^\text{large} \\ (u_{2i}, u_{2i+1})\in\{\eta(X),\zeta(X) \}}}z(X)\\
&+\diff(\rho_l, \Omega_\old\setminus \Omega_\old^\text{large}, z_\old, F_0) +\diff(\rho_l, \Omega\setminus\Omega^\text{large}, z, F)\\
& \geq (C-16)n^{2/3}-2 -\frac{1}{2}\sum_{i=1}^l\sum_{\substack{X\in\Omega_\old^\text{large} \\ (u_{2i-1}, u_{2i})\in\{\eta(X),\zeta(X) \}}}z_\old(X)-\frac{1}{2}\sum_{i=1}^{l-1}\sum_{\substack{X\in\Omega^\text{large} \\ (u_{2i}, u_{2i+1})\in\{\eta(X),\zeta(X) \}}}z(X) \\
&+\diff(\rho_l, \Omega_\old\setminus \Omega_\old^\text{large}, z_\old, F_0)
+\diff(\rho_l, \Omega\setminus\Omega^\text{large}, z, F)
\end{aligned}$$

We first argue that the last two terms satisfy $\diff(\rho_l, \Omega_\old\setminus \Omega_\old^\text{large}, z_\old, F_0) \ge 0$ and $\diff(\rho_l, \Omega\setminus\Omega^\text{large}, z, F)\geq0$.
\begin{itemize}
	\item For any $X\in\Omega_\old\setminus\Omega_\old^\text{large}$, since $\rho_l$ starts and ends with edges in $F_0$, by Lemma~\ref{diff}, $\diff(\rho_l, X, F_0)\geq 0$. Therefore, $\diff(\rho_l, \Omega_\old\setminus \Omega_\old^\text{large}, z_\old, F_0)\geq 0$.
	
	\item Consider any blossom $X\in\Omega\setminus\Omega^\text{large}$. As $(u_{2l-1}, u_{2l})$ is not a bad edge, it does not belong to $\gamma(X)\cup \{\eta(X) \}$ since it is not a matching edge in $F$ and it is not in $\gamma(X)\cup I(X)\setminus F$. Therefore, if $(u_1, u_2)$ also does not belong to $\gamma(X)\cup \{\eta(X) \}$, then by Lemma~\ref{diff} $\diff(\rho_l, \Omega, F)\geq 0$.
	
	Now consider the case where $(u_1, u_2)\in \gamma(X)\cup \{\eta(X) \}$. If $(u_1, u_2)\in\gamma(X)$, then it must be $\eta(X) = \NULL$ since $u_1$ is unsaturated. Therefore, applying Lemma~\ref{diff0}, we know $\diff(\rho_l, \Omega, F)$ cannot be equal to $-1$ since otherwise $\eta(X)\in\rho_l$ would not be null, or in other words $\diff(\rho_l, \Omega, F)\geq 0$. If $(u_1, u_2) = \eta(X)$, then $\eta(X)\notin F$, and thus we also know $\diff(\rho_l, \Omega, F)$ cannot be equal to $-1$ since otherwise Lemma~\ref{diff0} guarantees that $\eta(X)$ should be a matching edge in $F$. Either way, $\diff(\rho_l, \Omega, z, F)\geq 0$.
\end{itemize}

By these two bullets, we have
$$\begin{aligned}
\dy(u_{2l})\geq (C-16)n^{2/3}-2 -\frac{1}{2}\sum_{i=1}^l\sum_{\substack{X\in\Omega_\old^\text{large}\\ (u_{2i-1}, u_{2i})\in\{\eta(X),\zeta(X) \}}}z_\old(X) -\frac{1}{2}\sum_{i=1}^{l-1}\sum_{\substack{X\in\Omega^\text{large} \\ (u_{2i}, u_{2i+1})\in\{\eta(X),\zeta(X) \}}}z(X)
\end{aligned}$$

\begin{lemma}\label{aug-num}
	There are at least $\frac{1}{12}Cn^{2/3}$ augmenting walks $\rho\in\walks$ in which $\dy(u_{2l})\geq 0$.
\end{lemma}
\begin{proof}
	Recall the inequality
	$$\dy(u_{2l})\geq (C-16)n^{2/3}-2 -\frac{1}{2}\sum_{i=1}^l\sum_{\substack{X\in\Omega_\old^\text{large} \\ (u_{2i-1}, u_{2i})\in\{\eta(X),\zeta(X) \}}}z_\old(X) -\frac{1}{2}\sum_{i=1}^{l-1}\sum_{\substack{X\in\Omega^\text{large} \\ (u_{2i}, u_{2i+1})\in\{\eta(X),\zeta(X) \}}}z(X)$$
	
	By Lemma~\ref{z-bound} and Lemma~\ref{z1-bound}, $\frac{1}{2}\sum_{X\in\Omega_\old^\text{large}}z_\old(X) + \frac{1}{2}\sum_{X\in\Omega^\text{large}}z(X)\leq 3.5Cn^{4/3} + 7n^{2/3} < 4Cn^{4/3}$ for large constant $C$. Consider the value of the summation $$\frac{1}{2}\sum_{i=1}^l\sum_{\substack{X\in\Omega_\old^\text{large} \\ (u_{2i-1}, u_{2i})\in\{\eta(X),\zeta(X) \}}}z_\old(X) +\frac{1}{2}\sum_{i=1}^{l-1}\sum_{\substack{X\in\Omega^\text{large}\\ (u_{2i}, u_{2i+1})\in\{\eta(X),\zeta(X) \}}}z(X)$$
	
	Using Lemma~\ref{unique}, each $z_\old(X), z(X)$ appears in the above summation for at most one augmenting walk. Hence for at most $\frac{1}{6}Cn^{2/3}$ of the augmenting walks, the value of the summation is larger than $\frac{4Cn^{4/3}}{\frac{1}{6}Cn^{2/3}} = 24n^{2/3}$.	In other words, for the rest of $(\frac{1}{4}C-4)n^{2/3} - \frac{1}{6}Cn^{2/3} = (\frac{1}{12}C-4)n^{2/3}$ augmenting walks, when $C$ is a large enough constant, 
	$$\dy(u_{2l})\geq (C-16)n^{2/3} -2- 24n^{2/3}\geq 0$$
\end{proof}

For the rest, we are only interested in augmenting walks $\rho\in\walks$ under the requirement that $\dy(u_{2l})\geq 0$; by Lemma~\ref{aug-num}; there are at least $(\frac{1}{12}C-4)n^{2/3}$ many of them. As for other augmenting walks, prune them from $\walks$.

\subsection*{Phase 5}

Now, consider the suffix of an augmenting walk $\rho\in\walks$, $\langle u_{2l}, u_{2l+2}, \cdots, u_{2s}\rangle$, and for notational simplicity, revert its order and rename it $\langle v[1], v[2], \cdots, v[2k+1]\rangle$, where $v[1] = u_{2s}, v[2k+1] = u_{2l}$. So $\dy(v[1]) = -\ceil{Cn^{2/3}}$ and $\dy(v[{2k+1}])\geq 0$. 

\begin{defn}\label{adjacent}
	Consider the subsequence $v[2i_1+1], v[2i_2+1], \cdots, v[2i_r + 1] = v[2k+1]$ of all \textbf{good and original} vertices with odd indexes $2i_j+1, 1\leq j\leq r$ plus $u_{2l}$; for notational convenience, define $i_0 = -1$.
	
	Clearly $i_{j+1}-i_j\geq 3$ since there are at least $6$ edges between $v[2i_j+1]$ and $v[2i_{j+1}+1]$ on the augmenting walk. Call two original vertices $v[2i_j+1]$ and $v[2i_{j+1}+1]$ \textbf{adjacent} if the sub-walk between them does not contain any bad edges; notice that in this case it must be $i_{j+1} = i_j + 3$.
\end{defn}

\begin{lemma}\label{adj}
	$\sum_{j=0}^{r-1}(i_{j+1}-i_j-3)\leq 3n^{2/3}$.
\end{lemma}
\begin{proof}
	If $v[2i_j+1]$ and $v[2i_{j+1}+1]$ are not adjacent, then there is at least $(i_{j+1}-i_j-3)/3$ bad edge on the walk from $v[2i_j+1]$ to $v[2i_{j+1}+1]$. Since we restrict ourselves to augmenting walks with at most $n^{2/3}$ bad edges, the total sum should be bounded by $n^{2/3}$.
\end{proof}

\begin{lemma}
	For any $0\leq j<r$, 
	$$\begin{aligned}
	\dy(v[2i_{j+1}+1]) - \dy(v[2i_j+1])&\leq 8(i_{j+1}-i_j)+\frac{1}{2}\sum_{h=i_j+1}^{i_{j+1}}\sum_{\substack{X\in\Omega_\old^\text{large}\\ (v[2h-1], v[2h])\in\{\eta(X),\zeta(X) \}}}z_\old(X)\\
	&+ \frac{1}{2}\sum_{h=i_j+1}^{i_{j+1}-1}\sum_{\substack{X\in\Omega^\text{large}, (v[2h] \\ v[2h+1])\in\{\eta(X),\zeta(X) \}}}z(X)
	\end{aligned}$$
\end{lemma}
\begin{proof}
	Apply Lemma~\ref{odd} on the alternating walk of $\rho$ starting from an non-matching edge in $F$ (let us call it $\rho_j$) from $v[2i_j+1]$ to $v[2i_{j+1}+1]$, we have
	$$\begin{aligned}
	\dy(v[2i_{j+1}+1]) - \dy(v[2i_j+1])&\leq 8(i_{j+1}-i_j)+\frac{1}{2}\sum_{h=i_j+1}^{i_{j+1}}\sum_{\substack{X\in\Omega_\old^\text{large} \\ (v[2h-1], v[2h])\in\{\eta(X),\zeta(X) \}}}z_\old(X)\\
	&+\frac{1}{2}\sum_{h=i_j+1}^{i_{j+1}}\sum_{\substack{X\in\Omega^\text{large} \\ (v[2h], v[2h+1])\in\{\eta(X),\zeta(X) \}}}z(X)\\
	&-\diff(\rho_j, \Omega\setminus\Omega^\text{large}, z, F) - \diff(\rho_j, \Omega_\old\setminus\Omega_\old^\text{large}, z_\old, F_0)
	\end{aligned}$$
	So it suffices to prove $\diff(\rho_j, \Omega\setminus\Omega^\text{large}, z, F)$ and $\diff(\rho_j, \Omega_\old\setminus\Omega_\old^\text{large}, z_\old, F_0)$ are both non-negative. 
	
	First, consider any $X\in \Omega_\old\setminus \Omega_\old^\text{large}$. When $1\leq j< r$, both $v[2i_{j+1}+1]$ and $v[2i_j+1]$ are not incident on bad edges, and thus $\diff(\rho_j, X, F_0)\geq 0$ by Lemma~\ref{diff}. When $j = 0$, as the edge incident on $v[1]$ is a matching edge in $F_0$ which does not belong to $\gamma(X)\cup I(X)\setminus F_0$, so we can still apply Lemma~\ref{diff} to argue $\diff(\rho_0, X, F_0)\geq 0$.
	
	Second, consider any $X\in \Omega\setminus\Omega^{\text{large}}$. For $j\geq 1$, since both $v[2i_j+1], v[2i_{j+1}+1]$ are good vertices, using Lemma~\ref{diff} we know $\diff(\rho_j, X, F)\geq 0$. As for $\diff(\rho_0, X, F)$, we only need to worry about the case where the starting edge $(v[1], v[2])$ belongs to $\gamma(X)\cup \{\eta(X) \}$; otherwise again by Lemma~\ref{diff} we know $\diff(\rho_0, X, F)\geq 0$. Since by definition $v[2i_1+1]$ is a good vertex, thus $(v[2i_1], v[2i_1+1])\notin\gamma(X)\cup I(X)$, so if $\diff(\rho_0, X, F) = -1$, then by Lemma~\ref{diff0}, the only possibility is $\eta(X)\in \rho_0\cap F$. Therefore the unsaturated vertex $v[1]$ does not belong to $X$ as $\eta(X)\neq \NULL$. Since the starting edge $(v[1], v[2])$ belongs to $\gamma(X)\cup \{\eta(X) \}$ but $v[1]\notin X$, the first edge can only be $\eta(X)$ which is a matched edge in $F$, contradiction. Therefore, $\diff(\rho_0, X, F)\geq 0$.
\end{proof}

Next we need introduce a useful lemma that plays the key role in our proof.
\begin{lemma}\label{key}
	Let $a_1, a_2, \cdots, a_t$ be an arbitrary sequence of integers, and let $b>0$ be an integer. Then, there exists at least $\frac{1}{b+1}(a_t - a_1) - \sum_{i=1}^{t-1}\max\{a_{i+1}-a_i-b, 0\}$ different integers $q\in [a_1, a_t]$, with the property that $\exists i, a_i, a_{i+1}\in [q-b, q+b]$.
\end{lemma}
\begin{proof}
	Prove this by an induction on the sequence length. When $t = 1$, the statement is trivial. Now consider the inductive step. For general $t>1$, without loss of generality assume $a_t > a_1$. Then, find the smallest index $s\geq 2$ such that $a_s> a_1$. Apply induction on sequence $a_s, a_{s+1}, \cdots, a_t$, there are at least $\frac{1}{b+1}(a_t - a_s) -\sum_{i=s}^{t-1}\max\{a_{i+1}-a_i-b, 0 \}$ different satisfying values $q\in [a_s, a_t]$. Consider two cases.
	\begin{itemize}
		\item $a_s, a_{s-1}\in [a_1-b, a_1+b]$.
		
		In this case, including $q = a_1$, the total number of different $q$'s is at least $1+\frac{1}{b+1}(a_t - a_s) -\sum_{i=s}^{t-1}\max\{a_{i+1}-a_i-b, 0 \}\geq \frac{1}{b+1}(a_t - a_1) - \sum_{i=1}^{t-1}\max\{a_{i+1}-a_i-b, 0\}$.
		\item Not both of $a_s, a_{s-1}$ belong to $[a_1-b, a_1+b]$.
		
		In this case, by definition of index $s$, $a_s > a_1\geq a_{s-1}$. Hence, $a_s > a_1+b$ or $a_{s-1}<a_1-b$, so we always have $a_s - a_{s-1}\geq b+1$ Therefore, 
		$$a_s - a_{s-1}-b\geq \frac{1}{b+1}(a_s - a_{s-1})\geq\frac{1}{b+1}(a_s - a_1)$$
		and thus by the induction, the number of different satisfying $q$'s within $[a_s, a_t]$ is at least $$\frac{1}{b+1}(a_t - a_s) -\sum_{i=s}^{t-1}\max\{a_{i+1}-a_i-b, 0 \}\geq \frac{1}{b+1}(a_t - a_1) - \sum_{i=1}^{t-1}\max\{a_{i+1}-a_i-b, 0\}$$
	\end{itemize}
\end{proof}

For notational convenience, define:
$$\zsum(j) \overset{\text{def}}{=} \frac{1}{2}\sum_{h=i_j+1}^{i_{j+1}}\sum_{\substack{X\in\Omega_\old^{\text{large}} \\ (v[2h-1], v[2h])\in\{\eta(X),\zeta(X) \}}}z_\old(X) + \frac{1}{2}\sum_{h=i_j+1}^{i_{j+1}}\sum_{\substack{X\in\Omega^{\text{large}} \\ (v[2h], v[2h+1])\in\{\eta(X),\zeta(X) \}}}z(X)$$

Apply Lemma~\ref{key} by substituting the following parameters:
\begin{itemize}
	\item $b\leftarrow 24, t\leftarrow r+1$;
	\item $a_1\leftarrow \dy(v[1]), a_2\leftarrow \dy(v[2i_1+1]), \cdots, a_{t-1}\leftarrow \dy(v[2i_{r-1}+1]), a_t\leftarrow 0$.
\end{itemize}

Using Lemma~\ref{adj}, then the total number of different integers $q\in [\dy(v[1]), 0] = [-Cn^{2/3}, 0]$, such that $\exists 1\leq j<t, a_j, a_{j+1}\in [q-24, q+24]$ is at least
$$\frac{1}{25}Cn^{2/3} - 8\sum_{j=0}^r(i_{j+1}-i_j-3) - \sum_{j=0}^{r-1}\zsum(j) \geq (\frac{1}{25}C-24)n^{2/3} - \sum_{j=0}^{r-1}\zsum(j)$$

(Note that a positive upper bound can cover that case that $a_{i+1}-a_i-b<0$.) 
Next we need to argue that the above quantity $(\frac{1}{25}C-24)n^{2/3} - \sum_{j=0}^{r-1}\zsum(j)$ is large for most augmenting walks $\rho\in\walks$, and shortly we will be restricting our attention only on those augmenting walks where this quantity is large. More specifically, as we did before, by Lemma~\ref{z-bound} and Lemma~\ref{z1-bound}, $\frac{1}{2}\sum_{X\in\Omega_\old^\text{large}}z_\old(X) + \frac{1}{2}\sum_{X\in\Omega^\text{large}}z(X)\leq 3.5Cn^{4/3} + 7n^{2/3} < 4Cn^{4/3}$. Then, using Lemma~\ref{unique}, each $z_\old(X), z(X)$ contributes to at most one $\sum_{j=0}^{r-1}\zsum(j)$. Therefore, we can conclude that for at most $\frac{1}{24}Cn^{2/3}$ choices of augmenting walks, the corresponding summation $\sum_{j=0}^{r-1}\zsum(j)$ is larger or equal to $96n^{2/3}$; in other words, for the rest $|\walks| - \frac{1}{24}Cn^{2/3}\geq (\frac{1}{24}C - 4)n^{2/3}$ walks, the corresponding sum $\sum_{j=0}^{r-1}\zsum(j)$ is bounded by $96n^{2/3}$. So after the pruning step, we still have $|\walks|\geq (\frac{1}{24}C - 4)n^{2/3}$.

\subsection*{Phase 6}
Let us continue with our terminologies in the previous phase. We already know $\forall \rho\in \walks, \sum_{j=0}^{r-1}\zsum(j)\leq 96n^{2/3}$. So there are at least $$(\frac{1}{25}C-24)n^{2/3} - \sum_{j=0}^{r-1}\zsum(j)\geq (\frac{1}{25}C - 120)n^{2/3}$$ different integers $q\in [-Cn^{2/3}, 0]$ such that $\exists 1\leq j<t, a_j, a_{j+1}\in [q-24, q+24]$. We first need to remove some of the $q$'s that will not be useful for future arguments. The removal consists of two steps.

\begin{enumerate}[(1)]
\item Exclude all those $q$ from $[-Cn^{2/3}, -Cn^{2/3}+24]\cup [-24, 0]$, so there are still $(\frac{1}{25}C - 120)n^{2/3}-50$ many of these $q$'s. Since $\dy(v[2k+1])\geq 0$, for any of the rest $q\in [-Cn^{2/3}+25, -25]$, the corresponding index $j$ cannot be $1$ or $t-1$. Hence, there always exists $j$ such that $\dy(v[2i_j+1]), \dy(v[2i_{j+1}+1])\in [q-24, q+24]$.

\item Exclude all those $q$ such that $v[2i_j+1]$ and $v[2i_{j+1}+1]$ are not adjacent, and this lower bound would become $(\frac{1}{25}C-120)n^{2/3} - 50 - 49n^{2/3} = (\frac{1}{25}C-169)n^{2/3} - 50$, as there are at most $n^{2/3}$ bad edges in $\rho$. This is because, recalling the definition of ``adjacent'' from Definition~\ref{adjacent}, there is at least one bad edge on the sub-walk between $v[2i_j+1]$ and $v[2i_{j+1}+1]$ if they are not adjacent, and each such bad edge can invalidate at most $49$ different $q$'s.
\end{enumerate}

\begin{defn}\label{trap}
	We call an integer $q \in [-Cn^{2/3}+25, -25]$ traps $\rho$ at $j$ if the following conditions hold.
	\begin{enumerate}[(1)]
		\item Original vertices $v[2i_j+ 1], v[2i_{j+1} + 1]$ are adjacent.
		\item $y(v[2i_j+1]), y(v[2i_{j+1} + 1]) \in [q - 24, q + 24]$.
		\item $\zsum(j) = 0$.
	\end{enumerate}
\end{defn}
 
\begin{defn}
	For each $q\in [-Cn^{2/3}, 0]$, let $Y_q\subseteq V$ be the set of all original vertices $u$ such that $\dy(u) = q$. For any two original vertices $u, v\in V$, a walk \textbf{directly connects} $u, v$ if $(u, v) = e\in E$ and this walk goes through edges $(u, e_u), (e_u, e_v), (e_v, v)$ consecutively.
\end{defn}

\begin{lemma}\label{neck}
	Suppose an integer $q$ traps $\rho$. Then, $\rho$ directly connects two original vertices in $\bigcup_{h=q-24}^{q+24}Y_h\cup \bigcup_{h=-q-42}^{-q+40}Y_h$.
\end{lemma}
\begin{proof}
	By definition, $v[2i_j+1], v[2i_{j+1}+1]$ are adjacent, so $i_{j+1} = i_j + 3$. Consider the $6$-hop sub-walk from $v[2i_j+1]$ to $v[2i_{j+1}+1]$ which starts with an edge in $F_0$. This walk passes through three original vertices consecutively, and call them $v[2i_j+1], u, v[2i_j+7]\in V$.
	
	On the one hand, consider the $3$-hop sub-walk $\varrho_1$ from $v[2i_j+1]$ to $u$ that begins with an edge not in $F$. By Lemma~\ref{even} and $\zsum(j) = 0$, we have
	$$\begin{aligned}
	\dy(u) &\geq -\dy(v[2i_j+1]) - 18 + \diff(\varrho_1, \Omega_\old\setminus \Omega_\old^{\text{large}}, z_\old, F_0) + \diff(\varrho_1, \Omega\setminus\Omega^{\text{large}}, z, F)
	\end{aligned}$$
	
	By definition of being adjacent, this $3$-hop sub-walk $\varrho_1$ does not contain any bad edges, and so using Lemma~\ref{diff} we know both $\diff(\varrho_1, \Omega_\old\setminus \Omega_\old^{\text{large}}, z_\old, F_0)$  and $\diff(\varrho_1, \Omega\setminus\Omega^{\text{large}}, z, F)$ are nonnegative. So $\dy(u)\geq -\dy(v[2i_j+1]) - 18$.
	
	On the other hand, consider the $3$-hop sub-walk $\varrho_2$ from $u$ to $v[2i_{j+1}+1]$ that begins with an edge in $F$. As this $3$-hop sub-walk does not contain any bad edges, by Lemma~\ref{even1}, we have
	$$\dy(u)\leq -\dy(v[2i_{j+1}+1]) +16 - \diff(\varrho_2, \Omega_\old\setminus \Omega_\old^{\text{large}}, z_\old, F_0) - \diff(\varrho_2, \Omega\setminus\Omega^{\text{large}}, z, F)$$
	
	By definition of being adjacent, this $3$-hop sub-walk $\varrho_2$ does not contain any bad edges, and so using Lemma~\ref{diff} we know both $\diff(\varrho_2, \Omega_\old\setminus \Omega_\old^{\text{large}}, z_\old, F_0)$  and $\diff(\varrho_2, \Omega\setminus\Omega^{\text{large}}, z, F)$ are non-negative. $\dy(u)\leq -\dy(v[2i_{j+1}+1]) +16$.
	
	To sum up, $-q-42\leq -\dy(v[2i_j+1])-18\leq \dy(u)\leq -\dy(v[2i_{j+1}+1])+16\leq -q+40$, which concludes the proof.
\end{proof}

\begin{lemma}\label{q-num}
	For any augmenting walk $\rho\in\walks$, there are at least $\frac{1}{26}Cn^{2/3}$ different values of $q$ that traps $\rho$.
\end{lemma}
\begin{proof}
	We have already proved that there are $(\frac{1}{25}C-169)n^{2/3}-50$ many different $q$'s that satisfy properties (1) and (2) of Definition~\ref{trap}. For each such $q$, let $j_q$ be the index such that both $\dy(v[2i_{j_q}+1]), \dy(v[2i_{j_q+1}+1])\in [q-b, q+b]$. Then there are at most $2b+1 = 49$ different $q$'s that share the same $j_q$. So if we take a summation over all $q$ the value of $\zsum(j_q)$, we have: 

	$$\sum_{q}\zsum(j_q)\leq 49\times 96n^{2/3} = 4704n^{2/3}$$

	So there are at most $4704n^{2/3}$ possible values of $q$ such that $\zsum(j_q)>0$, and consequently there are $(\frac{1}{25}C-169)n^{2/3}-50 - 4704n^{2/3}>\frac{1}{26}Cn^{2/3}$ different $q$'s that traps $\rho$, when $C \geq 5\times 10^{6}$ is a sufficiently large constant.
\end{proof}

In the next lemma, we argue that for a large amount of integers $q\in [-Cn^{2/3}+b+1,-b-1]\subset[-Cn^{2/3}, -1]$, $q$ traps a large number of different augmenting walks.

\begin{lemma}\label{p-num}
	There exist at least $0.01Cn^{2/3}$ different $q\in [-Cn^{2/3},-1]$ that traps at least $\frac{1}{1000}Cn^{2/3}$ different augmenting walks in $\walks$.
\end{lemma}
\begin{proof}
	We prove it by contradiction. Assume there exist $k (k > 0.99 Cn^{2/3})$ different $q\in [-Cn^{2/3},-1]$ that traps less than $\frac{1}{1000}Cn^{2/3}$ different augmenting walks. Consider the total number of $q$ trapped by all the augmenting walks (if an integer $q$ is trapped by two different augmenting walks, it is counted twice). 

	On one hand, assume there are $t > (\frac{1}{24}C-4)n^{2/3} > \frac{1}{25}Cn^{2/3}$ augmenting walks in $\walks$, and by Lemma~\ref{q-num}, each of these augmenting walks is trapped by at least $\frac{1}{26}Cn^{2/3}$ different $q$'s from $[-Cn^{2/3}, -1]$. So the total number of trapped $q$ is more than $\frac{1}{26}Cn^{2/3}t$. 

	On the other hand, there are $k$ different $q$ that traps less than $\frac{1}{1000}Cn^{2/3}$ different augmenting walks. For the rest $Cn^{2/3}-k$ different $q\in [-Cn^{2/3},-1]$, each of them traps at most $t$ augmenting walks. The total number of trapped $q$ is $$\frac{k}{1000}Cn^{2/3} + (Cn^{2/3}-k)t \le \frac{0.99}{1000}C^2n^{4/3} + \frac{1}{100}Cn^{2/3}t \le \frac{1}{26}Cn^{2/3}t$$
	the last inequality holds since $t > \frac{1}{25}Cn^{2/3}$. It is a contradiction.

\end{proof}

By Lemma~\ref{p-num}, we can collect $0.01Cn^{2/3}$ different $q\in [-Cn^{2/3}, -1]$ that traps at least $\frac{1}{1000}Cn^{2/3}$ different augmenting walks in $\walks$. We argue there exists such a $q$ such that the size of $$Y(q)\overset{\text{def}}{=}\bigcup_{h=q-24}^{q+24}Y_h\cup \bigcup_{h=-q-42}^{-q+40}Y_h$$ is at most $\frac{13200}{C}n^{1/3}$. In fact, each original vertex $u$ can belong to at most $132$ different $Y(q)$, and so by the pigeon-hole principle there exists such a $q$ such that $|Y(q)|\leq 132n / 0.01Cn^{2/3} = \frac{13200}{C}n^{1/3}$. By Lemma~\ref{neck}, every augmenting walk in $\walks$ trapped by $q$ directly connects two vertices in $Y(q)$. By edge-disjointness of augmenting walks, every pair of $u, v\in Y(q)$ can be directly connected at most once, and hence $\frac{1}{1000}Cn^{2/3} < \left(\frac{13200}{C}n^{1/3}\right)^2$, which is a contradiction when $C$ is a sufficiently large constant.

\bibliographystyle{plain}
\bibliography{ref}

\end{document}